\documentclass[conference,compsoc]{IEEEtran}
\usepackage{blindtext}
\usepackage{adjustbox}
\usepackage{multirow}
\usepackage{cite}
\usepackage{amsmath,amssymb,amsfonts}
\usepackage{algorithm,algpseudocode}
\algrenewcommand\algorithmicindent{0.5em}
\usepackage{wrapfig,lipsum,booktabs}
\usepackage{subfigure}
\usepackage{float}  
\usepackage{graphicx}
\usepackage{textcomp}
\usepackage{xcolor}
\usepackage{graphicx}
\usepackage{subcaption}
\usepackage{amsthm}
\usepackage{enumitem}
\usepackage[normalem]{ulem}

\newtheorem{definition}{Definition}
\newtheorem{theorem}{Theorem}[section]
\newtheorem{lemma}{Lemma}
\newtheorem{assumption}{Assumption}

\usepackage{hyperref}
\usepackage[compact]{titlesec}

\let\svthefootnote\thefootnote
\newcommand\freefootnote[1]{%
  \let\thefootnote\relax%
  \footnotetext{#1}%
  \let\thefootnote\svthefootnote%
}

\makeatletter
\algnewcommand{\LineComment}[1]{\Statex \hskip\ALG@thistlm \(\triangleright\) #1}
\makeatother

\usepackage[mathlines,switch]{lineno}

\makeatletter

\let\old@ps@IEEEtitlepagestyle\ps@IEEEtitlepagestyle
\def\confheader#1{%
    \def\ps@IEEEtitlepagestyle{%
        \old@ps@IEEEtitlepagestyle%
        \def\@oddhead{\strut\hfill#1\hfill\strut}%
        \def\@evenhead{\strut\hfill#1\hfill\strut}%
    }%
    \ps@headings%
}
\makeatother

\begin{document}

\title{zkSTAR: A Zero Knowledge System for Time Series Attack Detection Enforcing Regulatory Compliance in Critical Infrastructure Networks}
\author{
    \IEEEauthorblockN{Paritosh Ramanan, Sathwik Yamana, H M Mohaimanul Islam, Abhiram Reddy Alugula}
    \IEEEauthorblockA{
        School of Industrial Engineering and Management,
        Oklahoma State University \\
        Stillwater, Oklahoma, USA \\
        \{paritosh.ramanan, syamana, h\_m\_mohaimanul.islam, aalugul\}@okstate.edu
    }
}


\maketitle

\begin{abstract}
Industrial control systems (ICS) form the operational backbone of critical infrastructure networks (CIN) such as power grids, water supply systems, and gas pipelines. As cyber threats to these systems escalate, regulatory agencies are imposing stricter compliance requirements to ensure system-wide security and reliability. A central challenge, however, is enabling regulators to verify the effectiveness of detection mechanisms without requiring utilities to disclose sensitive operational data. In this paper, we introduce zkSTAR, a zero-knowledge based cyberattack detection framework that leverages zk-SNARKs to enable regulatory compliance while delivering provable detection guarantees with complete data privacy. Our approach builds on established residual-based statistical hypothesis testing methods applied to state-space detection models. Specifically, we design a two-pronged zk-SNARK architecture that enforces (i) temporal consistency of the state-space dynamics and (ii) statistical consistency of the detection tests, enabling regulators to verify correctness and prevent suppression of alarms without visibility into utility-level data. We formally analyze the soundness and zero-knowledge properties of our framework and validate its practical feasibility through computational experiments on real-world ICS datasets. Our work demonstrates that zkSNARKs can provide a compliant, scalable, privacy-preserving alternative for detecting data-driven cyberattacks on ICS driven critical infrastructure networks.
\end{abstract} 



\section{Introduction}\label{sec:intro}
The convergence of information and operational technologies (IT–OT) through widespread IoT adoption has expanded the attack surface for data-driven cyber intrusions. Such attacks often target industrial control systems (ICSs), degrading asset performance and triggering cascading failures across critical infrastructure networks (CINs). While detection occurs locally at utilities, regulatory bodies such as Information Sharing and Analysis Centers (ISACs) aggregate alarms in real time to enhance situational awareness and reduce false positives \cite{dhs_oig2}. The effectiveness of these regulators, however, hinges on the integrity and reliability of local detection data. Recent real-world incidents \cite{langner2011stuxnet,jeffries2022cyber,salazar2024tale} have demonstrated that adversaries can compromise both IT and OT layers to obscure, delay, or disable detection of underlying anomalies. Such incidents expose a fundamental regulatory blind spot pertaining to verifying the integrity of reported detections without re-auditing raw operational data. Moreover, inconsistent or low-quality local insights can undermine the ISAC’s ability to construct a reliable network-wide situational awareness model, potentially causing unnecessary service disruptions \cite{dhs_oig2}. To address this gap, we introduce zkSTAR, a framework enabling publicly verifiable detection of data-driven ICS attacks in complete zero knowledge of local ICS state-space dynamics.

zkSTAR addresses a critical regulatory gap by enabling provable temporal and statistical consistency verification of ICS detections. Leveraging knowledge-soundness and provably consistent outcomes, it ensures that compromised utilities cannot falsify or suppress alarms even if their IT–OT stacks are breached. The framework enables verifiable regulatory oversight while preserving utility data confidentiality \cite{dhs_oig2}, eliminating the need for costly, privacy-intrusive audits that have long impeded secure information sharing \cite{nolan2015cybersecurity}. Additionally, by adopting an on-demand, lazy proof generation paradigm, zkSTAR ensures a low utility overhead as well.

Conventional ICS anomaly detection relies on state-space models that estimate expected sensor behavior and compute residuals to flag deviations from steady-state conditions \cite{li2020detection}. Statistical hypothesis tests on these residuals distinguish normal from abnormal operation, triggering alarms when observed measurements significantly diverge from model predictions \cite{li2020deep}. With existing approaches, regulators can verify such detection outcomes only through full data audits of ICS operations—requiring utilities to relinquish proprietary IoT datasets and raising severe privacy and data-governance concerns. Moreover, utilities must demonstrate both temporal consistency wherein successive state estimates are derived from prior ones; and statistical consistency in which alarms follow valid hypothesis test procedures. These audits are costly, slow, and incompatible with real-time oversight.


To overcome these limitations, zkSTAR delivers provable temporal and statistical consistency guarantees for state space driven statistical detection of data driven ICS attacks. At its core, zkSTAR is specifically designed to validate consistency claims for ICS systems with highly non-linear state space dynamics in complete zero knowledge. zkSTAR is based on the zkSNARK (zero-knowledge succinct, non-interactive argument of knowledge) paradigm providing publicly verifiable attack detection outcomes for a state space modeling paradigm in a trustworthy fashion. zkSTAR operates in an on-demand compliance mode, where regulators issue proof-generation requests to utilities whenever verification of local detections is required. We present a detailed overview of our contributions as follows:
\begin{itemize}[noitemsep, topsep=0pt, parsep=0pt, partopsep=0pt]
    \item We develop a zkSNARK based prediction framework that leverages an extended Kalman filter based model to provide provable detection guarantees for data driven attacks on utility ICS in complete zero-knowledge of local state space dynamics.
    
    \item We establish formal security guarantees for zkSNARK driven regulatory compliance with temporal and statistical consistency proofs of detection outcomes that also ensure resistance to attack-suppression under knowledge soundness property.  
  
    \item We design a zkSNARK based kernel decomposition scheme that provides the ability to prove temporal and statistical consistency claims in a scalable and computationally efficient manner.

    \item We detail the performance of zkSTAR with respect to detailed case studies involving real-world datasets using a REST-API driven containerized testbed. 
\end{itemize}

Overall, we show that zkSTAR enforces integrity of detection proofs under adversarial falsification, preventing utilities from reusing stale states or forging alarms—making it a security enforcement layer for regulatory trust.

\section{Related Works}\label{sec:related_works}
Attacks on industrial control systems (ICSs) that target IT layers can often be detected through network traffic monitoring~\cite{caselli2016specification,ye2004robustness}. However, data-driven attacks that manipulate sensor or control data can evade such approaches by simultaneously impacting both IT and OT layers~\cite{urbina2016limiting,kang2014cyber,drias2015taxonomy}, potentially causing cascading failures across critical infrastructure. More recently, there have been several attack vulnerabilities that have been discovered such as like Industroyer and Industroyer2 \cite{salazar2024tale}, Triton \cite{jeffries2022cyber} and Stuxnet \cite{langner2011stuxnet}. These vulnerabilities have demonstrated the feasibility of joint manipulation of control systems and monitoring data motivating the need for verifiable detection mechanisms.


To address some of these challenges, several \emph{model-based} mechanisms have been proposed~\cite{huang2018online,hoehn2016detection,rahman2017multi}. Among these, \emph{Kalman filter based} approaches are particularly robust and flexible~\cite{li2020deep,li2020detection,li2021degradation,ramanan2021blockchain,li2022online}, leveraging residuals from state-space models as statistical evidence for anomaly detection. These residual-based tests differentiate routine faults from malicious attacks~\cite{li2021degradation} and can support decentralized, network-wide alarm dissemination~\cite{ramanan2021blockchain}.

From the compliance perspective, however, it is essential for the regulator to receive timely notifications regarding attack detection at utility stakeholders in addition to verifying the integrity of attack detection so as to limit its network-wide impacts and reduce false alarm rates \cite{ramanan2021blockchain,dhs_oig2}. Therefore, regulatory compliance goals would ideally enable stakeholders to provide publicly verifiable detection outcomes that are accurate and trustworthy without exposing private operational data. Achieving such verifiable compliance under privacy constraints motivates the use of zero-knowledge proofs \cite{narayan2015verifiable}, enabling a prover to convince a verifier of correctness without revealing sensitive information.


Originating from foundational work by Goldwasser \emph{et~al.}~\cite{goldwasser2019knowledge} and Blum \emph{et~al.}~\cite{blum2019non}, zero-knowledge proofs evolved into succinct, non-interactive forms (zk-SNARKs) that support efficient verification~\cite{kilian1992note,micali2000computationally,gennaro2013quadratic,parno2016pinocchio,ben2013snarks,costello2015geppetto,ben2014succinct}. Recent applications in machine learning verify training integrity and model ownership~\cite{huang2022zkmlaas,wang2022ezdps,nguyen2022preserving,yang2023fedzkp}. While existing zk-SNARK frameworks enable general-purpose verifiable computation, none target the verification of dynamical system integrity or hypothesis-test correctness which are cornerstones of industrial attack detection. Recent works on verifiable differential privacy \cite{narayan2015verifiable} and secure multiparty analytics \cite{froelicher2020drynx} focus primarily on data confidentiality but not on verifiable detection correctness, a gap zkSTAR directly addresses. To alleviate this core critical gap, this paper introduces zkSTAR, a framework that enables provable, privacy-preserving verification of correct execution of Kalman filter based detection mechanisms, paving the way for regulatory compliance in CINs.




\section{Nonlinear State Space Modeling}\label{sec:ssm}
For developing zkSTAR, we leverage the extended Kalman Filter (EKF) state space models owing to their flexibility in handling highly non-linear spatial and temporal interdependencies among several ICS components \cite{coskun2017long}. Therefore, we characterize the EKF-based, utility-level attack detection model as a foundational component of zkSTAR for achieving regulatory compliance under zero knowledge. Within zkSTAR, we develop a state-space formulation capable of identifying attack-induced anomalies through locally executed statistical hypothesis tests.

\subsection{Extended Kalman Filter Model}
We present a brief overview of the EKF model in this section while a detailed description of the model has been provided in Appendix \ref{app:ekf_app}. Let {${g}, {h}$} denote the state transition and observation functions respectively, $x_t\in R^{m}$ represents the latent space embedding, {$u_t\in R^{m}$} represents the control action and $y_t\in R^{d}$ represents noisy sensor measurements from asset sensors. Additionally, let process and measurement noises follow multivariate normal distributions $N(0,Q_t)$ and $N(0,R_t)$ respectively. Now, we define our EKF model based on the following equations:
\begin{gather}
    \hat{x}_{t|t-1}=g(\hat{x}_{t-1|t-1},u_{t-1}) \label{eq:KLQG1},\\
    r_t=y_t-h(\hat{x}_{t|t-1}) \label{eq:KLQG2},\\
    P_{t|t-1} = G_tP_{t|t-1}G^T_t + Q_{t-1} \label{eq:KLQG3},\\
    S_t = (H_tP_{t|t-1}H^T_t+R_t) \label{eq:KLQG4},\\
    K_t = P_{t|t-1}H^{T}_{t}S_t^{-1} \label{eq:KLQG5}\\
    \hat{x}_{t|t}=\hat{x}_{t|t-1}+K_tr_t \label{eq:KLQG6}\\
    P_{t|t} = (I - K_tH_t)P_{t|t-1} \label{eq:KLQG7}
\end{gather}
In our EKF model, we have $P_{t|t-1},P_{t|t}$ represent the predicted and the updated covariance estimates respectively, while $S_t$ represents the residual covariance at $t$. The state transition and the observation matrices at $t$ given by $G_t, H_t$ respectively are computed using the Jacobians of $g,h$.
\subsection{Hypothesis Test for Anomaly Detection}\label{subsec:hyp_test}
In this paper we focus primarily on false data injection (FDI) attacks \cite{li2020deep} wherein an adversary artificially injects noise into sensor measurements as a means to induce catastrophic damage to critical infrastructure. We consider residual covariance $S_t = U_t \Sigma_t U_t^{T}$ coupled with residual measurements $r_t$ as a means to detect such attacks \cite{li2020detection}. We denote the standardized vector of PC scores as $\tau_t = (S_t)^{-1/2}(r_t) = U_t\Sigma_t^{-1/2}(r_t)$ such that $\tau_t \sim N(0,I)$ \cite{li2020detection,hwang2017chi}. This implies that $||\tau_t||^2_2 \sim \chi^2_p$, where $p<d$ is the degrees of freedom corresponding to the number of principal components used \cite{hwang2017chi}. 
\begin{equation}
\rho =
\begin{cases}
1, & \text{if } T_{\chi^2,t} = \|\tau_t\|^2_2 > \chi^2_{m,\alpha}, \\
0, & \text{otherwise.}
\end{cases}
\label{eq:alarm_condition}
\end{equation}
A generic hypothesis testing procedure applicable to sensor and residual monitoring has been discussed in greater detail in \cite{li2021degradation}. We define a measure of \emph{probabilistic similarity} to help establish the distinctness of two sensor measurement vectors as stated in Definition \ref{defn8}.
\begin{definition}[Probabilistic Similarity]\label{defn8}
Two distinct $\tilde{y}_{t+1}, y_{t+1}$ with corresponding test statistics $\tilde{T}_t=\verb|ChiTestStat|(\tilde{y}_{t+1};\alpha,m)$ and $T_t=\verb|ChiTestStat|(y_{t+1};\alpha,m)$ are deemed to be probabilistically similar if either $\tilde{T}_t > \chi^2_{p,\alpha},\ T_t > \chi^2_{p,\alpha}$ or $\tilde{T}_t \leq \chi^2_{p,\alpha},\ T_t \leq \chi^2_{p,\alpha}$.
\end{definition}
The concept of probabilistic similarity is particularly useful while handling issues for numerical precision issues wherein the sensor data measurement stored and processed in the ICS differs from the original value recorded by the sensors. We proceed towards characterizing the threat model and security guarantees.


\subsection{Threat Model and Security Guarantees}
Our threat model considers adversaries whose objective is to induce catastrophic physical damage through statistically significant manipulation of sensor measurements. Adversaries may compromise sensor hardware or manipulate IT and SCADA systems to suppress alarms, but are assumed to lack long-term persistent control over sensing and control hardware. Therefore, we exclude gradual degradation based attacks \cite{li2021degradation} requiring persistent asset control over weekly or monthly horizons. Considering our focus on FDI attacks, adversarial data manipulation is assumed to occur at a single point in the ICS hierarchy, ultimately propagating to the anomaly detector. Our assumptions on the nature of attack and the underlying EKF model are formally defined as per the following assumptions.

\begin{assumption}[Additive Perturbations]\label{as:as1}
The reported sensor measurement $y^{rep}_t$ presented to the $\chi^2$ anomaly detector relates to the true post-attack measurement $y^{real}_t$ by an additive perturbation $\Delta_y \in \mathbb{R}^d$, as follows
\[y^{rep}_t = y^{real}_t + \Delta_y\]
\end{assumption}
Assumption~\ref{as:as1} characterizes the discrepancy $\Delta_y$ arising from benign sources such as sensor calibration errors or adversarial manipulation of the data pipeline between the point of attack and the detector.

\begin{assumption}[Operational Significance]\label{as:as2}
An attack is deemed operationally significant if the residual induced by the true post-attack sensor measurement satisfies $\|r^{real}_t\| \geq \Gamma$, where $r^{real}_t = y^{real}_t - h(\hat{x}_{t|t-1})$ and $\Gamma > 0$.
\end{assumption}
Assumption~\ref{as:as2} formalizes a system-specific damage threshold $\Gamma$, below which residual perturbations are insufficient to induce material impact on ICS assets \cite{li2020deep,ramanan2021blockchain}.

\begin{assumption}[EKF Sensitivity]\label{as:as3}
The test statistics $T^{rep}_t$ and $T^{real}_t$, corresponding to $y^{rep}_t$ and $y^{real}_t$ respectively, satisfy
\[\Pr\Big[T^{rep}_t < \chi^2_{m,\alpha} \mid T^{real}_t > \chi^2_{m,\alpha}\Big] \leq \epsilon_{max}\]
\end{assumption}
Assumption~\ref{as:as3} bounds $\epsilon_{max}$, the probability that a reported measurement yields a detection outcome inconsistent with the true post-attack measurement. Subject to Assumptions~\ref{as:as1} and~\ref{as:as2}, a well-formed detector ensures $\epsilon_{max} \ll \alpha$, establishing that alarm suppression through data manipulation alone is a low-probability event; formal bounding of $\epsilon_{max}$ in terms of attack magnitude and EKF covariance structure is deferred to future work.

The security objective of zkSTAR is to guarantee that residual and state measurements are legitimate outcomes of a committed EKF model and that anomalies are detected under a consistent hypothesis testing framework. The security objective ensures that any operationally significant attack must trigger an alarm or violate the cryptographic guarantees of zkSTAR.

\section{Theoretical Foundations of zkSTAR}
Detecting attacks requires monitoring covariance matrices of the residuals to drive the underlying $\chi^2$ hypothesis tests. However, publicly validating alarms by an external entity such as a regulator would require the disclosure of residual vectors and covariance matrices that enable these hypothesis tests. Therefore, a critical regulatory challenge is the verification of the hypothesis test outcomes as reported by the utilities in the absence of their state space residual and covariance estimates. Zero-knowledge proofs form a viable alternative for solving these critical pain points by helping achieve public verification of utility level attack detection outcomes for regulatory purposes.
Therefore, in this paper, we employ zero knowledge proofs to enable regulatory compliance and verification. We concentrate on a utility-oriented subsystem represented by a state-space model as outlined in Section \ref{sec:ssm} that is meant to be captured in a zkSNARK framework. As a result, our methodology specifically leverages zk-SNARKs to assert temporal consistency and statistical consistency for utility-level $\chi^2$ hypothesis tests in zero-knowledge of local state space dynamics.

\subsection{Zero Knowledge Preliminaries}
We begin our discussion of zk-SNARK based regulatory compliance schemes for cyber attack detection by establishing five fundamental definitions. 
\begin{definition}[Hash Function]\label{defn2}
A hash function denoted by \verb|hash| can be defined such that it can consume arbitrary length inputs so as to map to fixed length outputs of size $n$
\[
\verb|hash| : \{0,1\}^* \rightarrow \{0,1\}^n
\]
\end{definition}
Definition \ref{defn2} portrays a hashing function used to ensure the integrity of state space estimates. We assume \texttt{hash} is a cryptographic hash function satisfying collision resistance and preimage resistance. Since EKF model estimates usually lie in a bounded, predictable domain, a plain hash commitment is susceptible to enumeration by an honest-but-curious verifier. Therefore, we use nonce-based hash commitments $\verb|Comm|(m) = \verb|hash|(m||N)$ where $m,N$ form the message and unique nonce strings respectively.


\begin{definition}[Pretrained State Space Model]\label{defn1}
The non-linear Kalman Filter based state space can be represented by a pretrained model $\mathcal{M}$ parametrized by $\theta$ such that the following conditions hold
\[
\Delta^{in}_t = (\mathcal{H}^{in}_t,\hat{x}_{t|t}, y_t)\text{, }\mathcal{M}(\Delta^{in}_t;\theta) \mapsto \Delta^{out}_t
\]
\[
\Delta^{out}_t = (\mathcal{H}^{out}_{t},\rho_t,\hat{x}_{t+1|t+1}, r_t, K_t, Q_t, R_t)
\]
\[
\mathcal{H}^{in}_t = \verb|Comm|(\hat{x}_{t|t}), \text{ and } \mathcal{H}^{out}_{t} = \verb|Comm|(\hat{x}_{t+1|t+1})
\]

\end{definition}
Definition \ref{defn1} helps define a pretrained LSTM based model which encapsulates the non-linear Kalman Filter based state space as defined in Equations \eqref{eq:KLQG1}-\eqref{eq:KLQG7}. More specifically, the inputs and outputs of $\mathcal{M}$ at time $t$  are represented by $\Delta^{in}_t \text{ and } \Delta^{out}_t$ respectively.

\begin{definition}[zk-SNARK Setup Phase]\label{defn3}
Given a standard security parameter $\lambda$, the zk-SNARK setup phase can be defined by function \verb|Setup| such that
\[
\verb|Setup|(1^\lambda,\mathcal{M},\theta)\mapsto (pk,vk)
\]
\end{definition}

Definition \ref{defn3} discusses the setup phase of the zk-SNARK proof mechanism that involves the generation of $pk$ and $vk$ that represent the proving and verification keys respectively. We assume that the setup phase needs to be executed once for each utility. Such assumptions are driven by the fact that a pretrained, high accuracy state-space model $\mathcal{M}$ does not require frequent updates to the weight parameter. However, we note that the setup function needs to be re-evaluated upon each update of the weight parameter $\theta$. 

\begin{definition}[zk-SNARK Proving Phase]\label{defn4}
The zk-SNARK proving phase is defined by function \verb|Prove| such that
\[
\verb|Prove|(pk,\theta,\mathcal{M},\Delta^{in}_t,\Delta^{out}_t)\mapsto \Pi_t
\]
\end{definition}
Definition \ref{defn4} details the proving phase of the zk-SNARK. The \verb|Prove| function consumes the proving key $pk$, $\mathcal{M}$ parametrized by $\theta$ as well as the inputs $\Delta^{in}_t$ and outputs $\Delta^{out}_t$ at time $t$ to generate a proof artifact $\Pi_t$. The \verb|Prove| function is executed for each timestep $t$ at the utility level.

\begin{definition}[zk-SNARK Verification Phase]\label{defn5}
The verification of zk-SNARK proof artifacts can be represented by the function \verb|Verify| such that
\[
\verb|Verify|(vk,\Pi_t,\Delta^{in}_t,\Delta^{out}_t)\mapsto \Phi_t
\]
\end{definition}
Definition \ref{defn5} presents the \verb|Verify| function that aims to enable regulators to verify reported proofs $\Pi_t$ based on the verification key $vk$ as well as inputs $\Delta^{in}_t$ and $\Delta^{out}_t$. The \verb|Verify| function provides output $\Phi_t \in \{0,1\}$ which is indicative of successful verification. 

zkSTAR relies on several system-level assumptions surrounding the zero-knowledge proving system. First, as a consequence of using the \texttt{ezkl}/Halo2 backend \cite{ezkl2024}, zkSTAR depends on elliptic-curve cryptographic primitives and CRS integrity, assuming the CRS is generated via a publicly verifiable ceremony per-utility. Second, the verifier is assumed honest-but-curious, and public proof artifact values are assumed not to violate utility state-space privacy \cite{doe}. Finally, communication channels are assumed to preserve integrity but not confidentiality, which can be addressed via transport-level encryption such as TLS.


Zero knowledge regulatory compliance for attack detection relies on two operational attributes pertaining to temporal and statistical consistency. We argue that these two attributes are necessary and sufficient conditions to enable zero-knowledge regulatory compliance verification. Our argument is based on the ability of the pretrained model $\mathcal{M}$ to recursively ensure that state space estimates are derived from previous predictions while being statistically consistent with expected residual covariance. In the following subsections, we formally state each operational attribute separately as a precursor for guaranteeing integrity of detection outcomes.


\subsection{zkSTAR Temporal Consistency Guarantees}
Temporal consistency (TC) pertains to the recursive relationship between \textit{a posteriori} state space estimates at $t$ and $t+1$ respectively. Temporal consistency constraint encoded in the non-linear Kalman filter model ensures that the posteriori estimate $\hat{x}_{t+1|t+1}$ is a derivative of $\hat{x}_{t|t}$. We state this formally as part of Definition \ref{defn6}.
\begin{definition}[Temporal Consistency]\label{defn6}
The state space system is considered to be \emph{temporally consistent} if for each consecutive time step pairs \(t,t+1\), the following holds.
\[
x^{out}_{t} = \hat{x}_{t|t-1} + K_tr_t \text{, }  \hat{x}_{t+1|t} = g(x^{in}_{t+1},u_{t}) \text{, }x^{out}_{t} = x^{in}_{t+1}
\]


\end{definition}
As a consequence of Definition \ref{defn6} we can make the temporal consistency claim presented in Theorem \ref{thm:TC}.

\begin{theorem}[zk-STAR TC Conditions]\label{thm:TC}
Temporal consistency can be asserted across every consecutive time step pairs $t,t+1$ given sensor data measurements $y_{t+1}$ if and only if the following conditions hold
\[
\mathcal{H}^{in}_{t+1} = \mathcal{H}^{out}_{t} \text{ and } \Phi_{t+1} = 1  
\]
\end{theorem}
\begin{proof}
Proof provided in Appendix \ref{app:TC_proof}
\end{proof}

Theorem \ref{thm:TC} shows that the integrity of temporal posteriori predictions from an EKF driven state space system can be established using a zkSNARK proof artifact if and only if the hash commitments of the inputs to the EKF model at each time step is equal to that of the previous EKF output; and the proof artifact itself can be verified successfully. In other words, Theorem \ref{thm:TC} demonstrates that if temporally inconsistent values are used for carrying out the EKF predictions, then the verification fails for at least one time step. Therefore, to establish temporal consistency, it is necessary and sufficient to check alignment of commitments and verification status' across successive time steps. 
%

\subsection{zkSTAR Statistical Consistency Guarantees}
We consider the local attack detection framework characterized by the $\chi^2$ hypothesis test presented in Section \ref{subsec:hyp_test}. The $\chi^2$ test relies on the residual and state space estimates provided by $\mathcal{M}$ that is parametrized by $\theta$. As a result of Equation \eqref{eq:alarm_condition}, the local attack detection framework is expected to yield an alarm value denoted by $\rho_t=1$, when an attack is taking place while maintaining $\rho_t=0$ under normal circumstances. In practice however, sensor measurements traversing a utility ICS are subject to calibration errors and finite-precision effects on inexpensive hardware, making numerically exact reproduction at every ICS layer unrealistic. The inherent robustness of hypothesis tests permits approximate agreement, since a measurement that is close but not identical to the originally recorded value is sufficient for SC to hold.

Therefore, we utilize fundamental statistical principles to establish the notion of statistical consistency (SC). For the purposes of SC, we denote the true post-attack alarm status and the actual outcomes of the $\chi^2$ anomaly detector as $\rho^{real}_t \in \{0,1\}$ and $\rho^{rep}_t \in \{0,1\}$ respectively.
\begin{definition}[Statistical Consistency]\label{defn7}
Attack detection is deemed to be \emph{statistically consistent} if $\rho^{real}_t = \rho^{rep}_t$.
\end{definition}
SC thus ensures that the reported alarm values faithfully reflect what a genuine, unmanipulated execution of the $\chi^2$-based detection framework would have deemed as the true outcome. Lemma~\ref{lem:LM1} formalizes SC by characterizing the probabilistic similarity condition on sensor measurements in alignment with our stated threat model.


\begin{lemma}\label{lem:LM1}
Probabilistically similar sensor measurement realizations $y^{real}_t,y^{rep}_t$ are  sufficient conditions for observing statistically consistent outcomes $\rho^{real}_t,\rho^{rep}_t$.
\end{lemma}
\begin{proof}
Proof provided in Appendix \ref{subsec:LM1_proof}
\end{proof}
Lemma \ref{lem:LM1} shows that as long as the post-attack sensor signal  is probabilistically similar to the values presented to the anomaly detector, the zkSTAR framework retains statistical consistency. Under our threat model, Lemma \ref{lem:LM1} permits the sensor data presented to zkSTAR and its true post attack counterparts to be non-identical, leading to identical alarm outcomes only under conditions of probabilistic similarities.

\begin{theorem}[zk-SNARK SC Conditions]\label{thm:SC}
Given consecutive time step pairs $t, t+1$ temporal consistency of the system operating on $y^{rep}_{t+1}$ is necessary and sufficient for statistical consistency with respect to $y^{real}_{t+1}$ to hold with probability at least $1 - \epsilon_{max} - \alpha$.
\end{theorem}
\begin{proof}
Proof provided in Appendix \ref{subsec:SC_proof}
\end{proof}
Theorem \ref{thm:SC} shows that statistical consistency is a consequence of temporal consistency claims that helps establish soundness of the alarms themselves. 

Intuitively, temporal consistency ensures the EKF runs on the true posterior state measurements. Assumption~\ref{as:as1} as defined in our threat model guarantees that the true post-attack measurement is probabilistically faithful to what the detector would have seen; and Lemma~\ref{lem:LM1} then establishes that the reported alarm must match the true outcome. Using the statistical and temporal guarantees, we characterize the ability of the zkSTAR framework to prevent suppression of detection outcomes in the subsequent section.

\subsection{Detection Suppression Guarantees}
We prove that zkSTAR achieves its security objectives by guaranteeing that any attempt to forge or reuse stale state estimates for regulatory verification violates the TC-SC consistency relation, thereby ensuring attack-suppression detection. For a given timestep $t$, we denote $\Omega_t$ to be the witness for corresponding the state space model $\mathcal{M}$ parametrized by $\theta$. Using Theorems \ref{thm:TC} and \ref{thm:SC}, the temporal-statistical consistency relation $\mathcal{R}$ can be defined as in Equation \eqref{eq:Rt}.
\begin{equation}\label{eq:Rt}
\begin{aligned}
\mathcal{R}\Big[(\Delta_t^{in},& \Delta_t^{out}), \Omega_t\Big]
= 
\Big[
(\mathcal{M}(\Delta_t^{in}; \theta) = \Delta_t^{out})\\
&\wedge(\mathcal{H}^{in}_{t+1} = \mathcal{H}^{out}_t)
\wedge
(T_{\chi^2,t}(\tau_t) \le \chi^2_{m,\alpha})
\Big] \in \{0,1\}.
\end{aligned}
\end{equation}
We consider a utility prover $\mathcal{U}$ undergoing a data-driven cyberattack starting at timestep $t^{attack}>>t^{init}$ where $t^{init}$ represents the initial time step pertaining to the setup phase (Definition \ref{defn1}). Concretely, we base our analysis on the Halo2 backend for establishing knowledge-soundness. We assume that the objective of the malicious attacker is to carry out the attack while suppressing attack detection. Formally, we define detection suppression using Definition \ref{def:suppression}.

\begin{definition}[Detection Suppression]\label{def:suppression}
Detection using $\mathcal{M}(.;\theta)$ is deemed to have been suppressed if for any $t > t^{attack}$ the following conditions hold
\[
\begin{aligned}
\verb|Verify|(vk, \Pi^{rep}_t, \Delta_t^{in}, \Delta_t^{out}) = \Phi^{rep}_t = 1\\
\mathcal{R}\Big[(\Delta_t^{in}, \Delta_t^{out}), \Omega^{rep}_t\Big] = \Psi^{rep}_t =  1
\end{aligned}
\]
\end{definition}
\noindent 
As a result using the TC and SC claims detailed in Theorems \ref{thm:TC} and \ref{thm:SC}, we characterize the infeasibility of detection suppression in Theorem \ref{thm:dsks}.

\begin{theorem}[Detection Suppression Bounds]\label{thm:dsks}
Given a probabilistic polynomial-time adversary $\mathcal{A}$ and an underlying proving system comprising knowledge-soundness, the following relation holds for any timestep $t > t^{attack}$
\[
\begin{aligned}
    \Pr\Big[\Psi^{rep}_t = 1|\Phi^{rep}_t=1\Big] \leq \epsilon_{max} +\alpha + \texttt{negl}(\lambda)
\end{aligned}
\]
\end{theorem}
\begin{proof}
Proof provided in Appendix \ref{subsec:proof_dsks}.
\end{proof}

Theorem \ref{thm:dsks} fundamentally demonstrates that a utility undergoing an attack must either fail proof verification or violate TC-SC relations enforced by zkSTAR. In other words, an under-attack utility cannot fabricate or reuse stale system states without detection, since any such attempt would contradict the knowledge-soundness of the Halo2 backend rendering attack suppression cryptographically infeasible. 


\section{Computational Architecture of zkSTAR}

zkSTAR comprises two components corresponding to the utility as prover, encapsulating the EKF model in a zero-knowledge framework, and the regulator as verifier, asserting TC, SC, and detection suppression claims. Proofs are generated in an \emph{on-demand, lazy fashion}, wherein utilities commit to the EKF model a priori and generate witnesses in real time, with regulators invoking proof generation for specific detection windows as needed.

Encapsulating the EKF model into a zkSNARK structure presents three challenges. First, asserting TC is challenging due to the highly non-linear and dynamic nature of utility based ICS subsystems. Second, ensuring statistically consistent detection outcomes without revealing parameters of the statistical test is difficult. Lastly, high-fidelity streaming sensor datasets possess a fine temporal granularity that leads to an additional layer of computational complexity when characterizing attacks. In order to address these challenges, the zkSTAR framework adopts a temporal decomposition scheme that specifically caters to TC and SC aspects. To address these, zkSTAR adopts a temporal decomposition scheme tailored to TC and SC aspects.

\subsection{Temporal Decomposition Scheme}\label{subsec:tds}
We consider a real-time attack detection horizon that is discretized into a sequence of detection windows denoted by set $\mathcal{W}_i, \forall i\in\{1,2,\ldots\}$. Each detection window $\mathcal{W}_i$ is in turn divided into a set of $W$ TC intervals. Each TC interval is represented by $w_{ij},\forall j\in\{1,W\}$ and is primarily geared towards enforcing TC across multiple prediction time steps. Every TC interval consists of $D$ consecutive timestamps such that $w_{ij}=\{t_{w_{ij}},t_{w_{ij}}+1, t_{w_{ij}}+2\ldots t_{w_{ij}}+D\}$, where $t_{w_{ij}}$ is the starting timestamp of TC interval $w_{ij}$. The relationship between each detection window $\mathcal{W}_i$, its constituent TC intervals $w_{ij}$ and corresponding timesteps can be compactly denoted using Equations \eqref{eq:detwin} and \eqref{eq:tcint}.
\begin{gather}
    \mathcal{W}_i = \{w_{i1},w_{i2},\ldots w_{iW}\}, \ \forall i\in 1,2,\ldots \label{eq:detwin}\\
    w_{ij}=\{t_{w_{ij}},t_{w_{ij}}+1, \ldots t_{w_{ij}}+D\}, \ \forall j\in \{1,W\} \label{eq:tcint}
\end{gather}
Without loss of generality, let us consider a detection window $\mathcal{W}$ and its constituent TC intervals $w_j \in \mathcal{W},\ j\in \{1,W\}$. We know from Section \ref{sec:ssm} that for any time step $t_w+k\in w$, $r_{t_w+k}\sim N(0,S_{t_w+k})$ where $r_{t_w+k},S_{t_w+k}$ are the residual and the covariance matrix predicted by the pre-trained EKF model for time step $t_{w}+k$. Since $r_{t_w+k} \sim N(0,S_{t_w+k})$, the residuals and covariance matrices can be aggregated across all TC intervals in a detection window to yield $r_W$ and $S_W$ denoted in Equation \eqref{eq:agg}.
\begin{gather}
r_W = \sum\limits_{j=1}^{W}\sum\limits_{k=1}^{D}r_{t_{w_j}+k} \ \text{,} \  S_W = \sum\limits_{j=1}^{W}\sum\limits_{k=1}^{D}S_{t_{w_j}+k}\label{eq:agg}
\end{gather}
As a result of hierarchical decomposition, there exists a distinct detection alarm $\rho_W$ for each detection window $\mathcal{W}$.

The temporal decomposition scheme paves the way for zkSNARKs to help enforce integrity of state space predictions based on a pre-committed EKF model. However, the zkSNARK circuit size grows exponentially with the growth in the detection window size making the proof generation step computationally expensive especially for larger systems. Therefore, in order to balance practicality and computational efficiency, we adopt a kernel based strategy that decomposes various compute tasks pertaining to an EKF model. Specifically, we design zkSNARK based kernels that target individual sub tasks pertaining to TC and SC attributes of EKF predictions and detection outcomes capable of operating on batched EKF residual outputs. We note that current zkSNARK frameworks such as \texttt{ezkl} lack support for linear solvers, SVD decompositions, and derivative evaluations of complex sub modules due to their high computational cost in zero knowledge \cite{ezkl2024}. To address these limitations, we employ a \emph{compute-offline–validate-online} strategy, performing such expensive operations off-circuit while verifying their consistency with the state-space model within the zkSNARK circuit.

\begin{figure*}[!htb]
    \centering
    \includegraphics[width=0.95\textwidth,trim={0 0 4cm 0},clip,keepaspectratio]{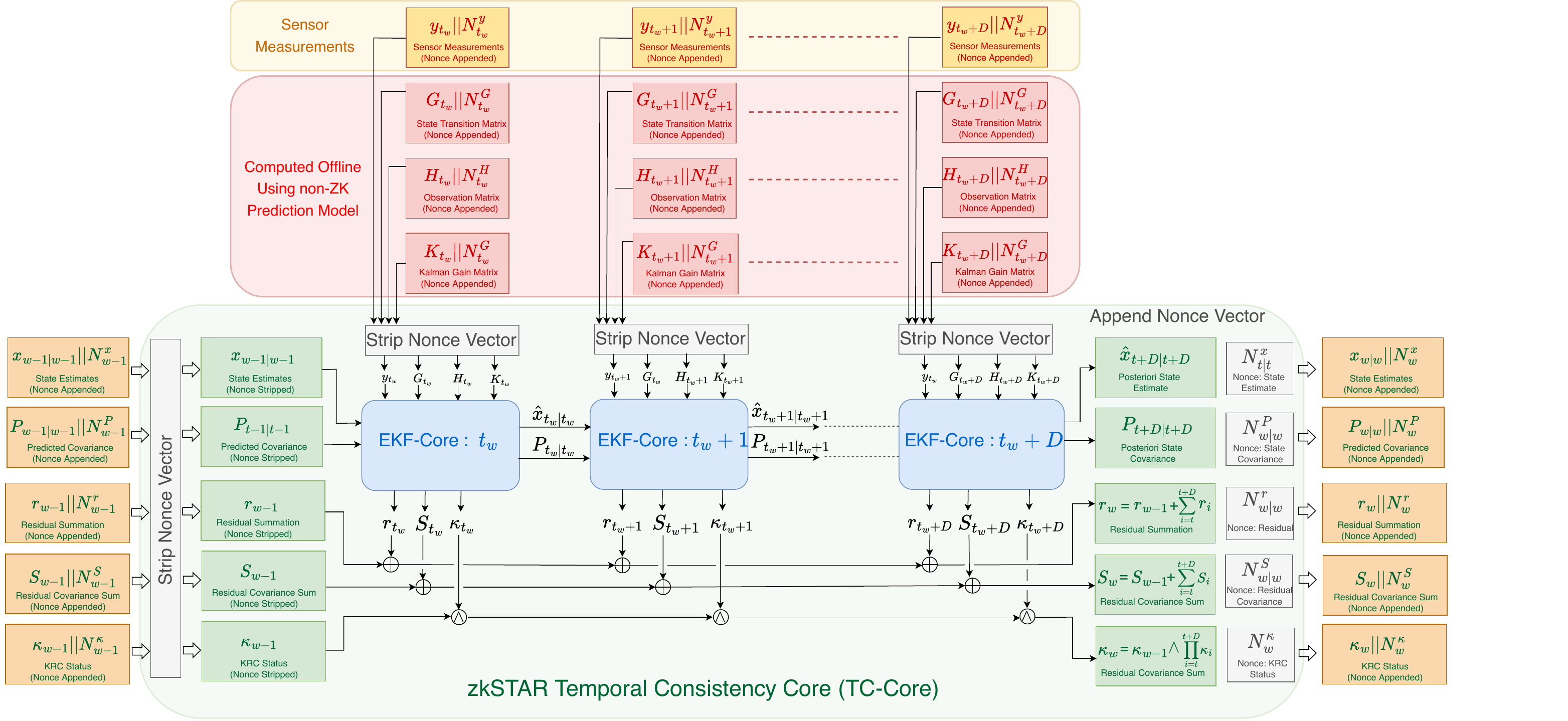}
  \caption{Temporal consistency handling in zkSTAR}\label{fig:temp_consist_zkSTAR}
\end{figure*}
\subsection{zkSNARK based Temporal Kernels}
For ensuring temporally consistent outcomes, our kernel design strategy involves breaking down distinct components of the EKF model inference. Specifically, we decompose the EKF inference over a single TC interval into distinct phases that pertain to validating the consistency of the Kalman gain matrix with the state covariance and process and measurement Jacobians; inference of the pre-committed EKF model on one timestamp of the TC interval; and accumulating the residual and its covariance sum across the entire TC interval.

\noindent
\textbf{Kalman Reconstruction Check (KRC) Kernel}: The EKF model presented in Section \ref{sec:ssm} relies on successive linear solves of Equation \eqref{eq:KLQG5} in order to determine the Kalman gain $K_t$. However, this poses a major challenge since implementing linear system solvers within a zkSNARK is computationally very expensive. Therefore, we develop a computational kernel tailored for the zkSNARK paradigm that demonstrates that offline computed matrices Kalman gain $K$, LSTM Jacobian $H$, residual covariance $S$ and state covariance $P$ are consistent with respect to the state space EKF model. The result is the KRC kernel depicted in Algorithm \ref{alg:kc_kernel} which yields a boolean output $\kappa_t$ depending on whether the error in Equation \eqref{eq:KLQG5} is below or above $\epsilon_{KC}$
\begin{algorithm}[htbp]
\caption{zkSTAR Kernel for Kalman Gain Reconstruction}
\label{alg:kc_kernel}
\begin{algorithmic}[1]
\Function{KRC-Kernel}{$K$, $S$, $H$, $P$, $\epsilon$}
    \State $\kappa \gets 0$
    \If{$\| K S - P H^\top \|_2^2 < \epsilon$} $\kappa \gets 1$
    \EndIf
    \State \Return $\kappa$
\EndFunction
\end{algorithmic}
\end{algorithm}

\noindent
\textbf{Extended Kalman Filter (EKF) Kernel}: For enforcing the state space dynamics at each timestamp in zero knowledge, we develop a dedicated zkSNARK based computational kernel for EKF denoted by $\mathcal{M}^{EKF}_t$.
\begin{gather}\label{eq:ekf_map}
   \mathcal{M}^{EKF} :
\underbrace{\mathcal{I}^{EKF}_{t} = \{\mathcal{X}_{t-1}, y_t, \mathcal{J}_{t}\}}_{\text{EKF input at timestamp $t$}}
\;\mapsto\;
\underbrace{\mathcal{O}^{EKF}_{t} = \{r_t,S_t,\kappa_t\}}_{\text{EKF output at timestamp $t$}} 
\end{gather}
Fundamentally, $\mathcal{M}^{EKF}$ forms a mapping between input tuple $\mathcal{I}^{EKF}_{t}$ and output tuple $\mathcal{O}^{EKF}_{t}$. The input tuple in turn consists of
$\mathcal{X}_{t-1} =\{x_{t-1|t-1},P_{t-1|t-1}\}$ which denotes a collection of state and covariance estimates from the prior time step. $\mathcal{J}_{t}=\{G_{t},H_{t},K_{t}\}$ represents a collection of the Jacobian matrices of the LSTM and feed-forward neural layers in the EKF formulation along with the Kalman gain matrix that are computed offline.  Lastly, $y_t$ represents the sensor data sampled at time stamp $t$. The output tuple $\mathcal{O}^{EKF}_{t}$ consists of $r_{t}$,$S_{t}$,$\kappa_{t}$, which denote the residual and its covariance as well as the status of the Kalman Gain reconstruction check. 
\begin{algorithm}[htbp]
 \caption{zkSTAR Kernel for Extended Kalman Filter}\label{alg:EKF_kernel}   
    \begin{algorithmic}[1]
        \Function{EKF-Kernel}{$x_{t-1|t-1}$,$P_{t-1|t-1}$,$y_{t}$,$G_{t}$,$H_{t}$,$K_{t}$}
            \State compute $\hat{x}_{t|t-1}$ using Equation \eqref{eq:KLQG1}.
            \State compute $r_{t}$ using Equation \eqref{eq:KLQG2}
            \State compute priori covariance $P_{t|t-1}$ using Equation \eqref{eq:KLQG3}
            \State compute innovation $S_{t}$ using Equation \eqref{eq:KLQG4}.
            \State $\kappa_{t_w}\leftarrow$\Call{KRCKernel}{$K_{t}$,$S_{t}$,$H_{t}$,$P_{t|t-1}$,$\epsilon_{KC}$}
            \State update posteriori estimate $\hat{x}_{t|t}$ using Equation \eqref{eq:KLQG6}
            \State update the state posteriori covariance $P_{t|t}$ using \eqref{eq:KLQG7}
            \Return $r_{t}$,$S_{t}$,$\kappa_{t}$
        \EndFunction
    \end{algorithmic}
\end{algorithm}
Algorithm \ref{alg:EKF_kernel} represents the computational kernel for EKF within the zkSTAR framework denoted by $\texttt{EKF-Kernel}$. For a timestamp $t$, $\texttt{EKF-Kernel}$ computes the priori state and covariance estimates $\hat{x}_{t|t-1}, P_{t|t-1}$ as well as the resulting residuals $r_t$ using Equations \eqref{eq:KLQG1}, \eqref{eq:KLQG3} and \eqref{eq:KLQG2} respectively. Using Equation \eqref{eq:KLQG4}, the kernel estimates the innovation or the residual covariance matrix $S_{t}$. Next, the $\texttt{KRC-Kernel}$ is invoked to validate the consistency of the Kalman gain matrix $K_t$ with respect to the Jacobian $H_t$, the residual covariance $S_t$, as well as the posteriori covariance $P_{t|t}$. Subsequently, the posteriori state and covariance estimates are updated according to Equations \eqref{eq:KLQG6} and \eqref{eq:KLQG7}. The kernel returns the current residual $r_t$, the residual covariance $S_t$ and the Kalman reconstruction status $\kappa_t$. 

\noindent
\textbf{Temporal Consistency (TC) Kernel}: While the EKF kernel can enforce state space dynamics for each timestamp $t$, it cannot independently ensure the integrity of state space estimates across multiple time stamps. In order to do so, we need a mechanism to assert the integrity between distinct timestamps. Therefore, we develop a temporal consistency (TC) kernel that is geared towards ensuring the integrity across distinct timestamps. Using the temporal decomposition scheme highlighted in Section \ref{subsec:tds} we focus on a single TC interval denoted by $w$ to aid our discussion. We denote the inputs to the TC kernel in terms of three sets: $\mathcal{X}^{TC}_{w-1}, \mathcal{R}^{TC}_{w-1}$ and $\mathcal{K}^{TC}_{w-1}$ which represent state, residual and KRC status estimates from the previous TC interval $w-1$. 
\begin{gather}
    \mathcal{X}^{TC}_{w-1} = \underbrace{\{(x_{w-1|w-1}||N^x_{w-1}), (P_{w-1|w-1}||N^P_{w-1})\}}_{\text{nonce appended state estimates}}\label{eq:xtc_wm1}\\
    \mathcal{R}^{TC}_{w-1} = \underbrace{(r_{w-1}||N^r_{w-1}), (S_{w-1}||N^S_{w-1})}_{\text{nonce appended residual and covariance}}\label{eq:rtc_wm1}\\
    \mathcal{K}^{TC}_{w-1} = \underbrace{(\kappa_{w-1}||N^{\kappa}_{w-1})}_{\text{nonce appended KRC status}}\label{eq:ktc_wm1}
\end{gather}
Each set in turn consists of the raw values that are appended with a unique randomly chosen nonce value to provide a robust degree of entropy to the resulting commitment strings. As a result, $\mathcal{X}^{TC}_{w-1}$ consists of $(x_{w-1|w-1}||N^x_{w-1}), (P_{w-1|w-1}||N^P_{w-1})$ where $x_{w-1|w-1}$, $P_{w-1|w-1}$, $N^x_{w-1}$ and $N^P_{w-1}$ denote the posteriori state and covariance estimate and their corresponding nonce values from the previous TC interval. Similarly, $\mathcal{R}^{TC}_{w-1}$ comprises the residual and covariance pairs $(r{w-1}||N^r_{w-1})$ and $(S_{w-1}||N^S_{w-1})$, where $r_{w-1}$ and $S_{w-1}$ represent the innovation and its associated covariance from the previous window, and $N^r_{w-1}$ and $N^S_{w-1}$ are the corresponding nonce values ensuring stochastic uniqueness of their hash commitments. Finally, $\mathcal{K}^{TC}_{w-1}$ encapsulates $(\kappa_{w-1}||N^{\kappa}_{w-1})$, where $\kappa_{w-1}$ denotes the binary KRC indicator and $N^{\kappa}_{w-1}$ is its appended nonce value used to prevent deterministic linkage across temporal proofs.
\begin{gather}
        \mathcal{Z}^{TC}_{t} = \{\underbrace{(y_t||N^y_t),(G_{t}||N^G_t),(H_{t}||N^H_t),(K_{t}||N^K_t)}_{\text{nonce appended sensor data, Jacobian and Kalman Gain}}\}\\
        \mathcal{Z}^{TC}_{w} = \{{\underbrace{{\mathcal{Z}^{TC}_{t_w},\mathcal{Z}^{TC}_{t_w+1},\ldots \mathcal{Z}^{TC}_{t_w+D}}}_{\text{EKF inputs for TC interval w}}}\}\label{eq:ztc_w}
\end{gather}
Additionally, the TC kernel also utilizes the set $\mathcal{Z}^{TC}_{w}$ which encapsulates the offline computed inputs and sensor data denoted by $\mathcal{Z}^{TC}_{t}$ meant for the EKF Kernel invocations at each timestep of the TC interval. The set $\mathcal{Z}^{TC}_{t}$ contains the nonce-appended sensor measurements and model parameters $(y_t||N^y_t)$, $(G_t||N^G_t)$, $(H_t||N^H_t)$, and $(K_t||N^K_t)$ corresponding to the sensor data, the Jacobians, and Kalman gain, respectively. Each nonce value used in $\mathcal{Z}^{TC}_{t}$ is generated through a cryptographically secure pseudorandom number generator enabling state linkability across successive EKF invocations. Collectively, the windowed set $\mathcal{Z}^{TC}_{w}$ aggregates these per-timestep inputs over the duration of the TC interval $w$, thereby serving as the complete batch of EKF inputs required for evaluating the zkSTAR TC kernel and generating the proof of temporal consistency for that interval. Consequently, Equation \eqref{eq:tc_map} is represents the TC kernel as a mapping denoted by $\mathcal{M}^{TC}$. 
\begin{equation}\label{eq:tc_map}
\begin{aligned}
\mathcal{M}^{TC} :
\;\underbrace{\mathcal{I}^{TC}_{t_w}
 = \{[\mathcal{X},\mathcal{R},\mathcal{K}]_{w-1},
    [\mathcal{Z},\mathcal{N}]_{w}\}}
    _{\text{\small input for TC interval $w$}}
\;\mapsto\;\\[3pt]
\qquad \underbrace{\mathcal{O}^{TC}_{t_w}
 = \{[\mathcal{X},\mathcal{R},\mathcal{K}]_{w}\}}
    _{\text{\small output of TC interval $w$}}.
\end{aligned}
\end{equation}
\begin{algorithm}[htbp]
 \caption{zkSTAR Temporal Consistency (TC) Kernel}\label{alg:tc_kernel}
    
    \begin{algorithmic}[1]
    \Function{TC-Kernel}{$\mathcal{X}^{TC}_{w-1}$,$\mathcal{R}^{TC}_{w-1}$,$\mathcal{Z}^{TC}_{w}$,$\mathcal{N}^{TC}_w$}
        \State {\text{\texttt{\#Strip nonce vectors from inputs}}}
        \State $\{x_{t_w-1|t-1},P_{t_w-1|t_w-1}\}\leftarrow$\Call{StripNonce}{$\mathcal{X}^{TC}_{w-1}$}
        \State $\{r_{w},S_{w},\kappa_{w}\}\leftarrow$\Call{StripNonce}{$\mathcal{R}^{TC}_{w-1}$} 
        \State $\{y_{t},G_{t},H_{t},K_{t}\}^{t_w+D}_{t_w}\leftarrow$\Call{StripNonce}{$\mathcal{Z}^{TC}_{t}$}
        \For{$t\in \{t_{w},t_{w}+1, t_{w}+2 \ldots t_{w}+D\}$}
            \State $\mathcal{X}_{t-1}\leftarrow\{x_{t-1|t-1},P_{t-1|t-1}\}$
            \State $\mathcal{J}_{t}\leftarrow\{G_{t},H_{t},K_{t}\}$
            \State $\mathcal{O}^{EKF}_{t}\leftarrow$\Call{EKF-Kernel}{$\mathcal{X}_{t-1}, y_t, \mathcal{J}_{t}$}
            \State accumulate residual sum $r_w\leftarrow r_w+r_t$
            \State accumulate residual covariance sum $S_w\leftarrow S_w+S_t$
            \State accumulate KRC status $\kappa_w\leftarrow \kappa_w\circ\kappa_t$
        \EndFor
        \State {\text{\texttt{\#Append nonce strings from $\mathcal{N}^{TC}_w$}}}
        \State $\mathcal{X}^{TC}_{w}\leftarrow${$\{(x_{t_w+D|t+D}||N^x_{w}),(P_{t_w+D|t_w+D}||N^P_{w})\}$}
        \State $\mathcal{R}^{TC}_{w}\leftarrow${$\{(r_{w}||N^r_{w}),(S_{w}||N^S_{w})\}$}
        \State $\mathcal{K}^{TC}_{w}\leftarrow${$(\kappa_{w}||N^\kappa_{w})$}
        \State \Return $\mathcal{X}^{TC}_{w},\mathcal{R}^{TC}_{w},\mathcal{K}^{TC}_{w}$
    \EndFunction
    \end{algorithmic}
\end{algorithm}
Algorithmically, the kernel is represented using the $\texttt{TC-Kernel}$ function presented in Algorithm \ref{alg:tc_kernel}. The kernel consumes the outputs of the previous TC interval $w-1$ and iterates through the timestamps $\{t_w\ldots t_{w}+D\}$ belonging to the current TC interval $w$. For each timestamp, the kernel strips the nonce vectors to obtain the underlying values which are fed to $\texttt{EKF-Kernel}$. The residual, covariance and KRC status outputs of the EKF kernel are accumulated across all the timestamps of the current TC interval $w$ through summation and multiplication operations respectively. At each iteration, the kernel returns the overall outputs pertaining to the state, residual and reconstruction status check for the corresponding TC interval.
\begin{figure}[!htb]
    \centering
\includegraphics[width=0.49\textwidth,keepaspectratio]{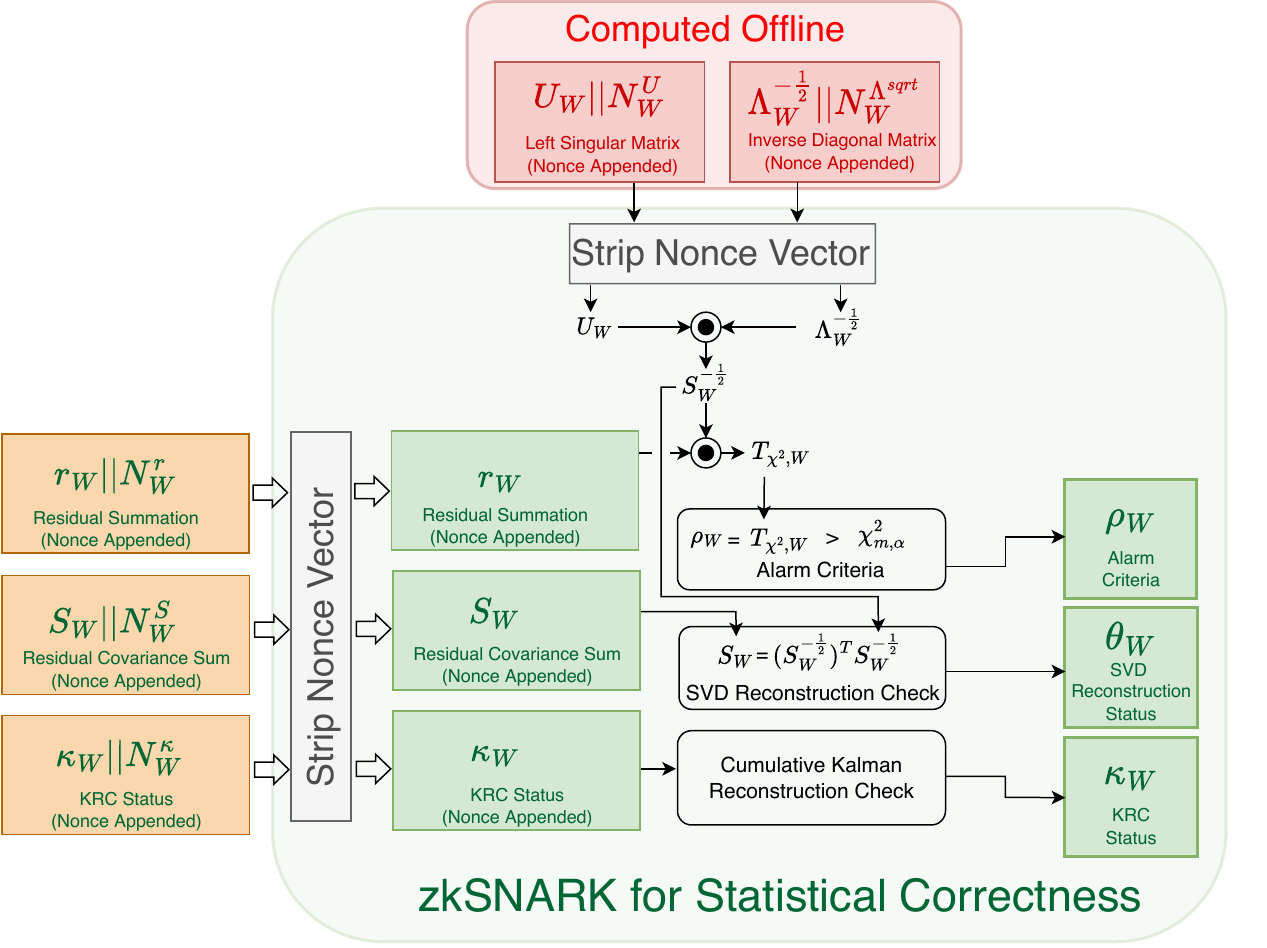}
  \caption{Statistical consistency estimation in zkSTAR}\label{fig:stat_correct_zkSTAR}
\end{figure}
\begin{algorithm}[htbp]
\caption{zkSTAR Kernel for SVD Reconstruction}
\label{alg:svd_kernel}
\begin{algorithmic}[1]
\Function{SVD-Kernel}{$U$, $\Sigma^{-\frac{1}{2}}$,$S$, $\epsilon$}
    \State $\eta \gets 0$
    \If{$\|(U\Sigma^{-\frac{1}{2}})(U\Sigma^{-\frac{1}{2}})^TS - I\|_2^2 < \epsilon$} $\eta \gets 1$
    \EndIf
    \State \Return $\eta$
\EndFunction
\end{algorithmic}
\end{algorithm}
\subsection{zkSNARK based Statistical Kernels}
Applying the $\chi^2$ test on the accumulated residuals and covariance requires us to carry out SVD decomposition to evaluate the test statistic and the alarm criteria in complete zero-knowledge. Therefore, we design a sequence of statistical zkSNARK kernels to help achieve efficiency in zkSTAR.

\noindent
\textbf{SVD Reconstruction Kernel}: We develop a kernel to validate the SVD decomposition of the accumulated covariance matrix in zero-knowledge. The SVD step is essential to compute the normalized PC scores for determining the value of the test statistic in conjunction with the accumulated residual values for each detection window. The SVD validation kernel \texttt{SVD-Kernel} is presented in Algorithm \ref{alg:svd_kernel}.   

Specifically the SVD kernel consumes the offline computed value of the left singular matrix $U$ as well as the square root of the singular values $\Sigma^{-\frac{1}{2}}$, an error tolerance threshold $\epsilon$ in addition to the accumulated covariance matrix $S$ itself. The kernel returns the validation status $\eta$ depending on whether the SVD reconstruction represented by $U(\Sigma^{-\frac{1}{2}})^T(U\Sigma^{-\frac{1}{2}})$ falls within an $\epsilon$ ball of the predicted covariance matrix accumulated across the detection window.

\noindent
\textbf{Hypothesis Test Kernel}: We conduct the $\chi^2$ hypothesis test using the \texttt{SC-Kernel} which considers the accumulated residual and covariance values at the end of every detection window. Additionally, we compute the SVD of the accumulated covariance matrix in an offline manner and validate its consistency using the SVD reconstruction kernel denoted by \texttt{SVD-Kernel}. We represent the inputs to the SVD kernel in terms of $\mathcal{Z}^{SVD}_{W}$ as stated in Equation \eqref{eq:svd_inputs} which consists of the singular vectors and the square roots of the singular values.
\begin{gather}\label{eq:svd_inputs}
        \mathcal{Z}^{SVD}_{W} = \{\underbrace{(U_W||N^U_W),(\Sigma^{-\frac{1}{2}}_{W}||N^\Sigma_W)}_{\text{nonce appended SVD values}}\}
\end{gather}
The hypothesis testing kernel pertaining to detection window W can therefore be denoted as a mapping $\mathcal{M}^{SC}_W$ as stated in Equation \eqref{eq:ht_map} which maps input set $\mathcal{I}^{SC}_{W}$ to output $\mathcal{O}^{SC}_{t_w}$.
\begin{gather}\label{eq:ht_map}
\mathcal{M}^{SC} :
\underbrace{\mathcal{I}^{SC}_{W} = \{\mathcal{R}_W, \mathcal{Z}^{SVD}_{W}\}}_{\text{input for HT}}\mapsto\underbrace{\mathcal{O}^{SC}_{t_w} = \{\rho_W, \eta_W, \kappa_W\}}_{\text{output of HT}}
\end{gather}
In Equation \eqref{eq:ht_map}, $\mathcal{R}_W$ represents outputs of \texttt{TC-Kernel} obtained upon the culmination of the current detection window while $\mathcal{Z}^{SVD}_W$ represents the offline computed SVD of the accumulated residual covariance matrix. 
\begin{algorithm}[htbp]
\caption{zkSTAR Kernel for the $\chi^2$ Hypothesis Test}
\label{alg:ht_kernel}
\begin{algorithmic}[1]
\Function{SC-Kernel}{$\mathcal{R}_W$,$\mathcal{Z}^{SVD}_W$}
    \State {\text{\texttt{\#Strip nonce vectors from inputs}}}
    \State $\{r_W,S_W,\kappa_W\}\leftarrow$\Call{StripNonce}{$\mathcal{R}_{W}$}
    \State $\{U_W,\Sigma^{-\frac{1}{2}}_W\}\leftarrow$\Call{StripNonce}{$\mathcal{Z}^{SVD}_{W}$}
    \State {\text{\texttt{\#Set UCL and default alarm for W}}}
    \State $T^{UCL}_W\leftarrow \chi^2_{m,\alpha}$, $\rho_W\leftarrow0$
    \State $\eta_W\leftarrow$\Call{SVD-Kernel}{$U_W$,$\Sigma^{-\frac{1}{2}}_W$,$S_W$, $\epsilon_{SVD}$}
    \State compute $T_{\chi^2,W}\leftarrow ||U_W\Sigma^{-\frac{1}{2}}_Wr_W||^2_2$ 
    \If{$ T_{\chi^2,W} > \chi^2_{m,\alpha}$ then $\rho_W \gets 1$}
    \EndIf
    \State \Return $\rho_W$, $\eta_W$, $\kappa_W$ 
\EndFunction
\end{algorithmic}
\end{algorithm}

In \texttt{SC-Kernel} represented in Algorithm \ref{alg:ht_kernel}, we consume $\mathcal{R}_W$ and $\mathcal{Z}^{SVD}_W$ and initiate the hypothesis testing procedure by stripping the nonce vectors from the inputs. We compute the test statistic $T_{\chi^2,W}$ by utilizing the offline computed SVD values. Additionally, the value of the upper control limit $T^{UCL}_{W}$ is baked into the proving circuit for individual subsystems since it is assumed to remain invariant across time, while the alarm value $\rho_W$ is set to 0 by default. We compare the test statistic with $T^{UCL}$ to determine the value of $\rho_W$ within the kernel. Lastly, the SVD kernel is invoked to ensure the consistency of the offline computed SVD values with the accumulated sensor covariance matrix $S_W$ with a threshold of $\epsilon_{SVD}$ to yield a consistency status denoted by $\eta_W$. The kernel outputs a tuple of values consisting of the alarm $\rho_W$, the SVD status check $\eta_W$ as well as the cumulative KCR status check of $\kappa_W$. 
\begin{algorithm}[!htb]
\caption{zkSTAR Algorithmic Framework}
\label{alg:global_algorithm}
\begin{algorithmic}[1]
\State {\text{\texttt{\#Set visibility for TC circuit}}}
\State {\text{\texttt{\#input;output;param visibility}}}
\State \texttt{hash;hash;private}$\leftarrow$\Call{ProvingMode}{$\mathcal{M}^{TC}$}
\State
\State {\text{\texttt{\#Setup proving circuit for TC}}}
\State \verb|Setup|($1^\lambda,\mathcal{M}^{TC},\theta^{TC})\mapsto (pk^{TC},vk^{TC}$)
\State
\State {\text{\texttt{\#Set visibility for SC circuit}}}
\State {\text{\texttt{\#input;output;param visibility}}}
\State \texttt{hash;public;private}$\leftarrow$\Call{ProvingMode}{$\mathcal{M}^{SC}$}
\State
\State {\text{\texttt{\#Setup proving circuit for SC}}}
\State \verb|Setup|($1^\lambda,\mathcal{M}^{SC},\theta^{TC})\mapsto (pk^{SC},vk^{SC}$)
\State
\State distribute $vk^{TC},vk^{SC}$ publicly
\State
\For{$i = 1,2\ldots$} \Comment{rolling detection windows $\mathcal{W}_i$}
    \For{$j=1,2,\ldots W$} \Comment{TC interval $w_{ij}$}
        \State generate nonce values $\mathcal{N}^{TC}_{j}$
        \State compute $\mathcal{Z}^{TC}_{j}$ using Equation \eqref{eq:ztc_w}
        \State compute $[\mathcal{X}, \mathcal{R}, \mathcal{K}]_{j-1}$ using Equations \eqref{eq:xtc_wm1}-\eqref{eq:ktc_wm1}.
        \State $\mathcal{I}^{TC}_{j}\leftarrow\{[\mathcal{X}, \mathcal{R}, \mathcal{K}]_{j-1},[\mathcal{Z}^{TC},\mathcal{N}^{TC}]_{j}\}$
        \State $\mathcal{O}^{TC}_{j}\leftarrow$\Call{TC-Kernel}{$\mathcal{I}^{TC}_{j}$}
        \State
        \State {\text{\texttt{\#Generate TC proof for interval $j$}}}
        \State $\Pi^{TC}_t \leftarrow\verb|Prove|(pk^{TC},\theta^{TC},\mathcal{M}^{TC},\mathcal{I}^{TC}_{j},\mathcal{O}^{TC}_{j})$
        \State
        \State extract $[\mathcal{X}, \mathcal{R}, \mathcal{K}]_{j}$ from $\mathcal{O}^{TC}_{j}$
    \EndFor
    \State extract $S_W$ from $\mathcal{R}_W$
    \State $U_W,\Sigma_W^{-\frac{1}{2}}\leftarrow$\Call{SVD}{$S_W$}\Comment{computed offline}
    \State generate nonce values $\mathcal{N}^{SVD}$
    \State compute $\mathcal{Z}^{SVD}_{W}$ using SVD values $U_W,\Sigma_W^{-\frac{1}{2}},\mathcal{N}^{SVD}$
    \State $[\rho, \eta, \kappa]_W \leftarrow$\Call{SC-Kernel}{$\mathcal{R}_W$,$\mathcal{Z}^{SVD}_W$} 
    \State
    \State {\text{\texttt{\#Generate SC proof for window $\mathcal{W}_i$}}}
    \State $\Pi^{SC}_t \leftarrow\verb|Prove|(pk^{SC},\theta^{SC},\mathcal{M}^{SC},\mathcal{I}^{SC}_{W},\mathcal{O}^{SC}_{W})$
    \State
    \State \Return $\mathcal{P_{W_i}} = \Big[\Big( \Pi^{TC}_{w_{i1}}\ldots\Pi^{TC}_{w_{iW}} \Big),\Pi^{SC}_{\mathcal{W_i}}\Big]$
\EndFor
\end{algorithmic}
\end{algorithm}
\subsection{zkSTAR Algorithmic Framework}
We describe the overall zkSTAR algorithmic framework in Algorithm \ref{alg:global_algorithm} that encapsulates the TC and SC kernels. The framework primarily executes on a rolling horizon of detection windows characterized by $\mathcal{W}_i$ with each window consisting of $W$ TC intervals of $D$ timestamps each. As a result, zkSTAR produces $W$ TC and 1 SC proof artifact for each detection window $\mathcal{W}_i$ which is denoted by the set $\mathcal{P_{W}}_i$ described in Equation \eqref{eq:proof_artifacts}.
\begin{gather}\label{eq:proof_artifacts}
    \mathcal{P_{W}}_i = \Big[\Big( \Pi^{TC}_{w_{i1}}\ldots\Pi^{TC}_{w_{iW}} \Big),\Pi^{SC}_{\mathcal{W_i}}\Big]
\end{gather}

zkSTAR performs a one-time setup for the temporal (TC) and statistical (SC) zkSNARK circuits. The TC circuit uses a \emph{hash, hash, private} visibility mode, ensuring inputs and outputs are committed as hashes $\verb|hash|(\mathcal{I}^{TC}_w)$, $\verb|hash|(\mathcal{O}^{TC}_w)$ while keeping model parameters private, yielding keys $(pk^{TC}, vk^{TC})$. The SC circuit adopts \emph{hash, public, private} visibility, committing inputs while exposing outputs $\mathcal{O}^{SC}_\mathcal{W}$ publicly, yielding $(pk^{SC}, vk^{SC})$. Both verification keys are published for regulator use. At runtime, zkSTAR iterates over detection windows and TC intervals per Algorithm~\ref{alg:global_algorithm}, invoking the TC-Kernel and SC-Kernel to produce verifiable zero-knowledge proofs of temporal and statistical consistency. For each detection window $\mathcal{W}_i$, the algorithm iterates through TC intervals $w_{ij}$ of length $W$. At each interval, fresh nonces $\mathcal{N}^{TC}_{j}$ are generated and used to construct the nonce-appended input set $\mathcal{Z}^{TC}_{j}$ per Equation~\eqref{eq:ztc_w}, encapsulating sensor data, Jacobians, and Kalman gain parameters. Together with the propagated state, residual, and reconstruction variables $[\mathcal{X},\mathcal{R},\mathcal{K}]_{j-1}$ from Equations~\eqref{eq:xtc_wm1}--\eqref{eq:ktc_wm1}, these form the TC-Kernel input $\mathcal{I}^{TC}_{j}$, whose output $\mathcal{O}^{TC}_{j} = [\mathcal{X},\mathcal{R},\mathcal{K}]_{j}$ is extracted and propagated to the next interval. After all $W$ TC intervals are processed, $\mathcal{R}_W$ yields the composite covariance $S_W$, whose offline SVD produces $U_W$ and $\Sigma_W^{-\frac{1}{2}}$. These are nonce-augmented into $\mathcal{Z}^{SVD}_W$, which together with $\mathcal{R}_W$ is consumed by the \textsc{SC-Kernel} to produce decision metrics $[\rho,\eta,\kappa]_W$ which are the zkSNARK-backed anomaly indicator for window $\mathcal{W}_i$.

\section{Experimental Results}
\begin{figure*}[!htb]
    \centering
    \includegraphics[width=\textwidth,keepaspectratio]{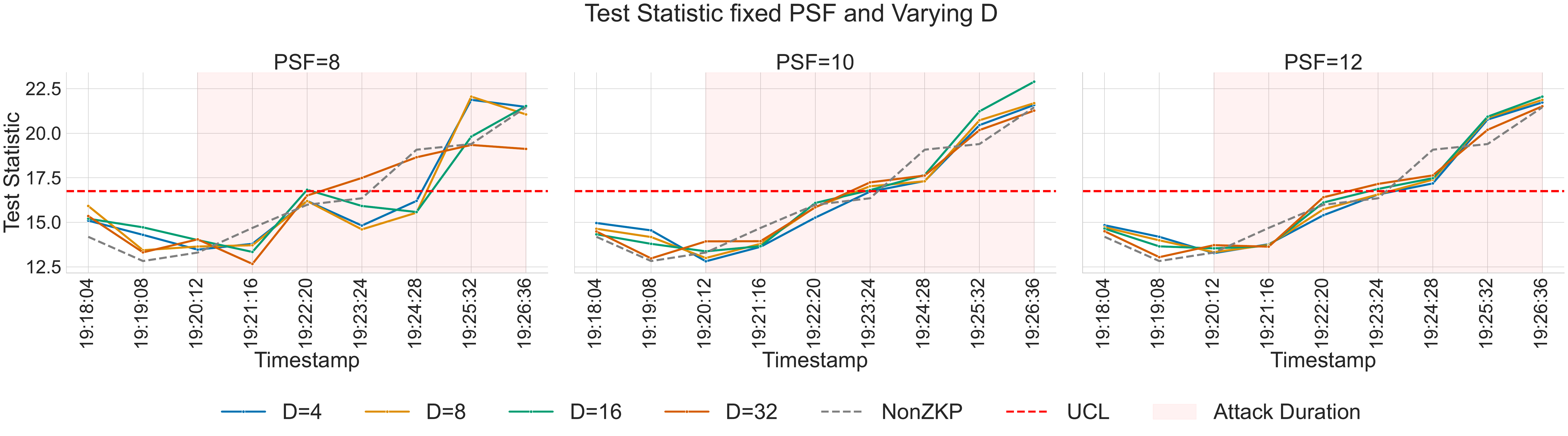}
  \caption{HAI Dataset: Detection Quality with varying D and fixed PSF}\label{fig:P1_dval}
    \vspace{-2mm}
\end{figure*}
\begin{figure*}[!htb]
    \centering
        \includegraphics[width=\textwidth,keepaspectratio]{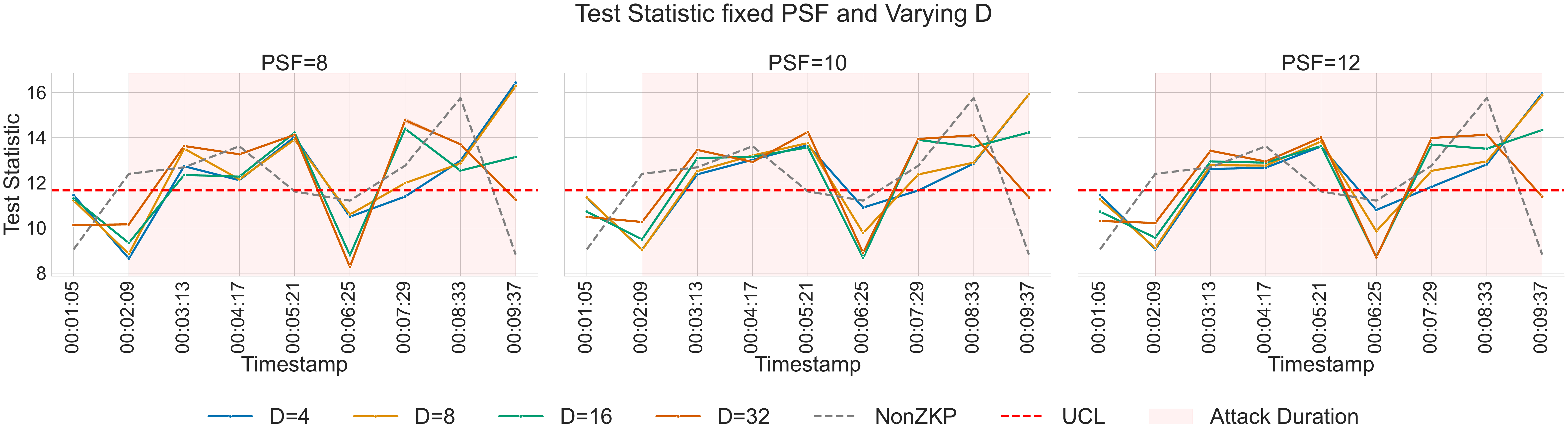}
  \caption{ORNL-PS Dataset: Detection Quality with varying D and fixed PSF}\label{fig:ps_dval}
  \vspace{-4mm}
\end{figure*}
For experimental studies, we utilize the HAI \cite{hai} and ORNL-PS \cite{pan2015classification,pan2015developing} datasets. The training methodology for EKF models is briefly discussed in Appendix \ref{subsec:ekf_train}. We carry out detection based on a pre-determined value of upper control limit (UCL) that is representative of the levels of significance used for the $\chi2$ hypothesis test. These UCL values are generally dependent on the subsystem data and might also rely on engineering domain expertise as well. We use a value of $5e-3$ and $2e-2$ for HAI and ORNL-PS datasets that were obtained after rigorous empirical studies. All of our results include comparison with the non-ZKP version as our benchmark to judge the detection quality.

\noindent
\textbf{System Implementation Details}: All experiments were carried out on a virtual machine running Ubuntu 24.04 with 100GB of RAM and 16 vCPUs, using \texttt{Python 3.11} and \texttt{PyTorch 2.7.1} for detection model inference. zkSTAR is implemented in both native and containerized execution modes using the \texttt{ezkl} library \cite{ezkl2024} for all zkSNARK generation. The zkSTAR code base, implementation details and instructions are provided here \cite{anonym2025zkSTAR}. The native mode targets high-performance computing clusters to evaluate detection accuracy and computational correctness across varying residual aggregation windows ($D$) and precision scale factors (PSF). The containerized mode assesses real-world CPU and memory usage during witness generation and proof generation and verification for TC and SC. Separate containers emulate the utility and regulator roles, communicating via REST APIs managed through the \texttt{zkstarctl} tool.

\subsection{Detection Quality of zkSNARK outcomes}
To evaluate detection quality, we conduct two experiments examining the effects of varying the PSF with fixed D and vice versa. For brevity, PSF variation results are shown in Appendix~\ref{subsec:dq_psf}. Each experiment spans a 10-minute detection window, with attacks initiated mid-run. For HAI, the attack window includes the actual 164s event and a 3-minute buffer to account for the non-linear model's delayed anomaly convergence, whereas in ORNL-PS, the attack persists until the end of the run. We test PSF values of 8, 10, and 12, and $D$ values of 4, 8, 16, and 32.


Figures \ref{fig:P1_dval} and \ref{fig:ps_dval} demonstrate the detection quality with varying D and while fixing the PSF values to 8, 10 and 12 respectively. From both figures we can observe that the detection quality with respect to the UCL remains robust across varying values of D and PSF. Generally, we can observe from both figures that the curves for different D values become more clustered around the non-zkSNARK version with increasing PSF values. This trend remains consistent across both datasets as well. 
In each case we can observe that the detection consistently occurs within the designated attack duration window providing similar quality as the non-ZKSNARK version. 
\begin{table*}[!htb]
\setlength{\tabcolsep}{5pt}
\centering
\begin{tabular}{|l|c|cc|cc|cc|cc|}
\hline
\textbf{PSF} & \textbf{Hypothesis Test} &
\multicolumn{2}{c|}{\textbf{D=1}} &
\multicolumn{2}{c|}{\textbf{D=4}} &
\multicolumn{2}{c|}{\textbf{D=8}} &
\multicolumn{2}{c|}{\textbf{D=16}} \\
\cline{2-10}
\textbf{Value} & \textbf{Time (s)} &
\textbf{Time (s)} & \textbf{Speedup} &
\textbf{Time (s)} & \textbf{Speedup} &
\textbf{Time (s)} & \textbf{Speedup} &
\textbf{Time (s)} & \textbf{Speedup} \\
\hline
8	&	56.06(22.60)	 &70.05(3.03)	 & 1.00 & 	80.81(4.19)	 & 3.5 & 	70.61(2.87)	 & 7.93 & 	78.25(3.75)	 & 14.32 \\ 	
10	&	54.26(19.17)	 &69.26(3.22)	 & 1.00 & 	71.32(3.22)	 & 3.89 & 	75.97(8.88)	 & 7.3 & 	77.11(5.01)	 & 14.37 \\ 	
12	&	53.49(6.60)	 &69.89(3.77)	 & 1.00 & 	66.50(2.34)	 & 4.20 & 	74.59(3.19)	 & 7.5 & 	77.79(3.90)	 & 14.37 \\ 
\hline
\end{tabular}
\caption{HAI Dataset: Speedup and Proof Generation Times (mean (std) in seconds)}
\label{tab:p1_pgt}
\end{table*}
\begin{table*}[!htb]
\setlength{\tabcolsep}{5pt}
\centering
\begin{tabular}{|l|c|cc|cc|cc|cc|}
\hline
\textbf{PSF} & \textbf{Hypothesis Test} &
\multicolumn{2}{c|}{\textbf{D=1}} &
\multicolumn{2}{c|}{\textbf{D=4}} &
\multicolumn{2}{c|}{\textbf{D=8}} &
\multicolumn{2}{c|}{\textbf{D=16}} \\
\cline{2-10}
\textbf{Value} & \textbf{Time (s)} &
\textbf{Time (s)} & \textbf{Speedup} &
\textbf{Time (s)} & \textbf{Speedup} &
\textbf{Time (s)} & \textbf{Speedup} &
\textbf{Time (s)} & \textbf{Speedup} \\
\hline
8  & 49.01 (9.68) & 81.74 (2.85) & 1.00 & 57.75 (1.35) & 5.66 & 67.69 (1.74) & 9.66 & 72.64 (1.77) & 18.0 \\
10 & 50.79 (1.32) & 69.82 (2.29) & 1.00 & 65.06 (1.17) & 4.29 & 68.66 (1.75) & 8.13 & 73.28 (1.83) & 15.24 \\
12 & 52.40 (37.36) & 69.27 (3.99) & 1.00 & 64.45 (1.26) & 4.29 & 67.21 (1.22) & 8.24 & 73.32 (2.17) & 15.11 \\
\hline
\end{tabular}
\caption{ORNL-PS Dataset: Speedup and Proof Generation Times (mean (std) in seconds)}
\label{tab:ps_pgt}
\end{table*}
\subsection{Computational Correctness of zkSTAR output}
The SVD and Kalman gain decompositions in the zkSTAR framework are computed offline. As a result, zkSTAR outputs reconstruction error, the associated threshold and a compute check flag for the SVD and Kalman gain matrices. We define the computational correctness as the absolute error emanating from the in-circuit reconstruction of the SVD and Kalman gain matrices computed with respect to the circuit encoded thresholds. Therefore, these experiments demonstrate the variation in relative absolute errors of all matrix elements with respect to the given threshold. In our experiments the threshold values for Kalman and SVD checks were set to 0.1. Figure \ref{fig:svd_cc} shows box plots of relative error for varying 
D values on the HAI and ORNL-PS datasets (PSF = 12) based on the SVD of predicted covariance matrices. For both datasets, relative errors remain below 1 in most cases, indicating reconstruction errors under 0.1. Slight error spikes above 1 occur during attack windows due to increased residual volatility, yet even in the worst case, relative errors stay below 1.4 (HAI) and 1.2 (ORNL-PS). Figure \ref{fig:kcr_cc} shows scatter plots of relative errors across varying \(D\) values. 
Unlike the SVD check, the Kalman gain check, which is performed every \(D\) steps, exhibits an initial staggered trend. 
For the HAI dataset, errors remain consistently low around \(1e-5\), while for the ORNL-PS dataset they are slightly higher yet stable, around \(1e-2\).

\subsection{System Performance of Utility Prover}
We begin by analyzing the system performance of zkSTAR proving system. Specifically, we examine the witness, proof, key and circuit sizes in addition to analyzing the utility's system performance of witness proof generation.
We analyze the computational performance  for temporal and statistical consistency proof generation for both datasets for varying D values. Circuit generation and setup costs are presented separately in Appendix \ref{subsec:cir_setup_cost}. Our results are presented in Tables \ref{tab:p1_pgt} and \ref{tab:ps_pgt} for HAI and ORNL-PS datasets respectively. The data for these experiments are obtained from multiple runs of varying PSF values and represent the mean and standard deviation for one proof invocation of the circuit. We also provide a speed up value that represents the average case speedup that can be expected when sequentially generating proofs for the entire horizon. We provide the speedup formula and its associated explanation using Equation \eqref{eq:speedupcalc} in Appendix \ref{subsec:pgt_app_subsec} respectively. 
\begin{figure}[!htb]
    \centering        \includegraphics[width=0.48\textwidth,keepaspectratio]{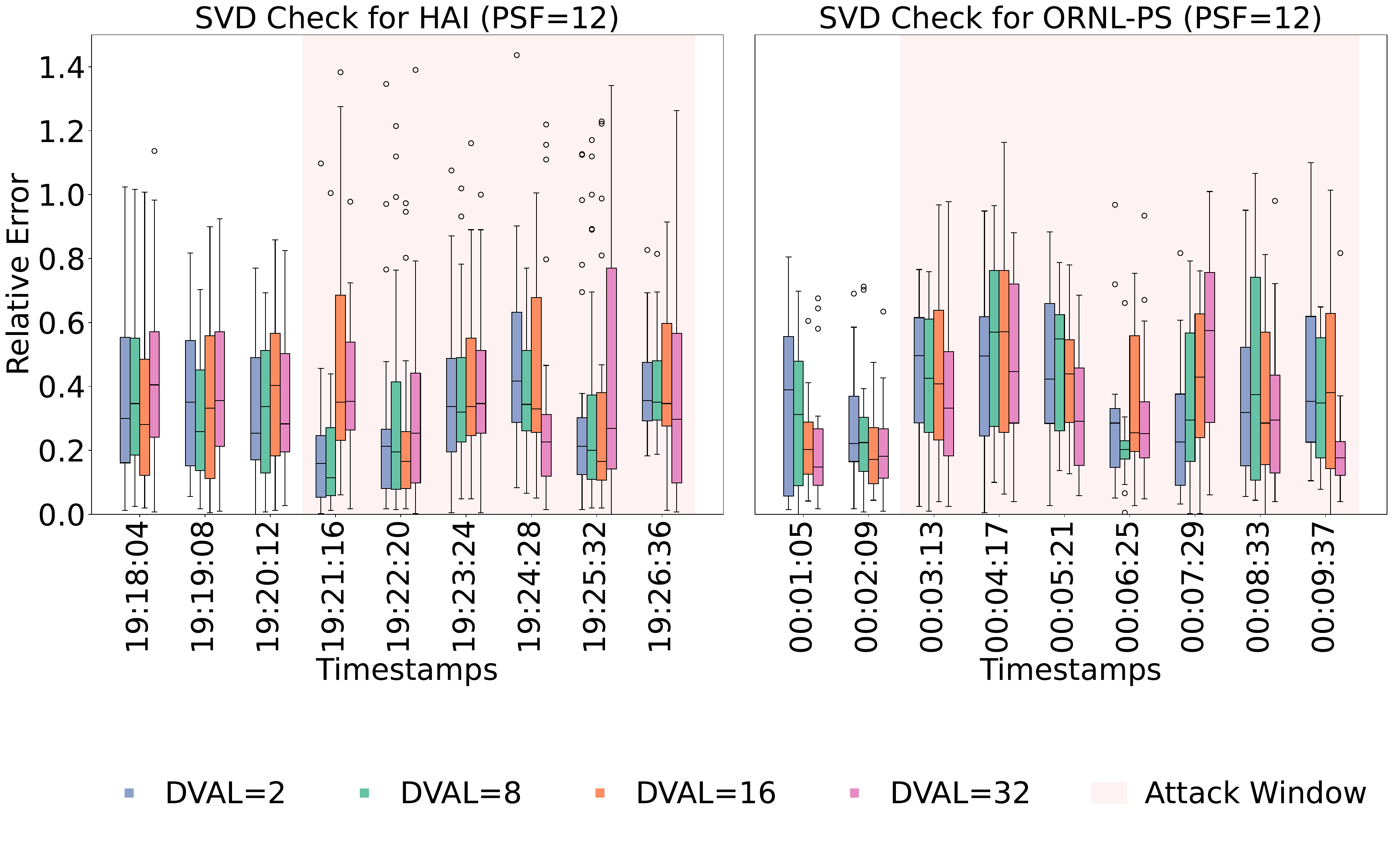}
  \caption{Computational error check for SVD of Kalman gain and aggregated covariance matrices}\label{fig:svd_cc}
  \vspace{-4mm}
\end{figure}
The temporal consistency proof generation times are represented by the mean and standard deviation values for D values of 1,4,8 and 16. The proving times are in general stable with minor variations observed with varying D values. We observe that in case of HAI, the mean proof generation time increases slightly going from D values of 1 to 4 before decreasing slightly from 4 to 8 and finally increasing from 8 to 16. On the other hand the trend is inverted for ORNL-PS, decreasing slightly from 1 to 4, increasing from 4 to 8 and then decreasing again from 8 to 16. 
\begin{figure}[!htb]
    \centering
\includegraphics[width=0.48\textwidth,keepaspectratio]{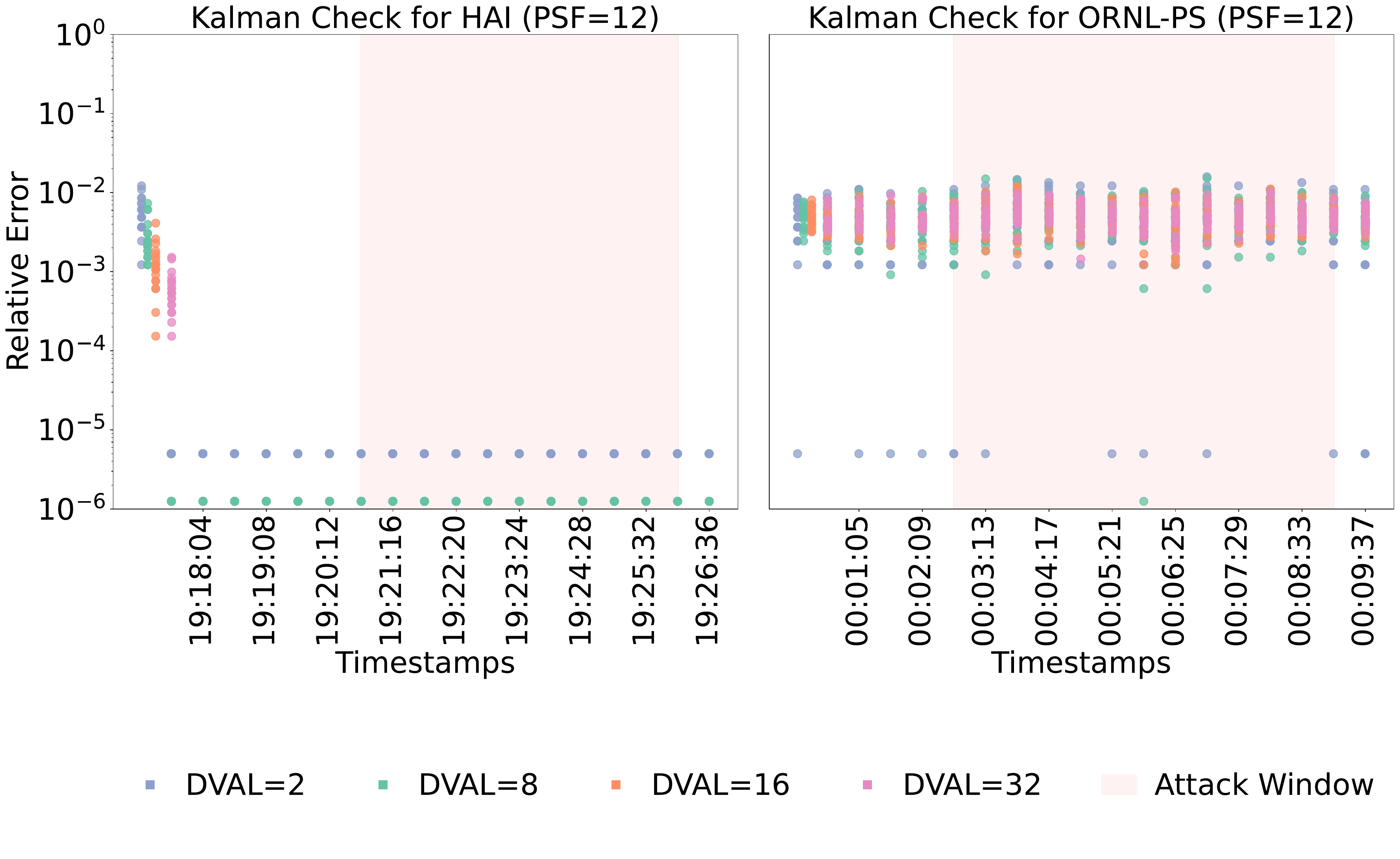} 
  \caption{Computational error check for SVD of Kalman gain and aggregated covariance matrices}\label{fig:kcr_cc}
\vspace{-2mm}
\end{figure}
Interestingly, there is more volatility for PSF value of 8 in case of either dataset as compared to 10 and 12. The non-monotonous trends depicted in Tables \ref{tab:p1_pgt} and \ref{tab:ps_pgt} can be due to the non-linear proving costs afforded by \texttt{ezkl} backend which could be explained on the basis of row domain padding used by Halo2 proving systems. Further, analyzing speedups, we note a healthy increase between each set of D values regardless of the PSF value or the dataset used. Since the circuit size is known to influence the proving times, the trends in Tables \ref{tab:p1_pgt} and \ref{tab:ps_pgt} are exactly aligned with those depicted in Table \ref{tab:tcws}. The speedup value insights point to the fact that having a higher TC interval size (or D value) is almost always more computationally efficient than invoking proof generation on a finer granularity corresponding to a lower TC interval size.


Lastly, the hypothesis test columns in Tables \ref{tab:p1_pgt} and \ref{tab:ps_pgt} reflect a relatively stable values of average proving time corresponding to statistical consistency proofs. These values are slightly higher for HAI compared to ORNL-PS which might be explained on the basis of the difference in dimensionality between both problems.



To gauge the regulatory overhead in terms of verification, we measure the verification times incurred by the regulator for verifying TC and SC proof artifacts. In almost all cases, zkSTAR SC and TC proofs are successfully verified in under 1.5s for both datasets with varying TC interval sizes. Detailed results pertaining to verification times for each case has been presented in Appendix \ref{subsec:verification_times}. These results assert the practicality of zkSTAR for near real-time regulatory audits with minimal overhead. 



\section{Conclusion}
In this paper, we address the problem of regulatory compliance in data-driven cyberattack detection for utility stakeholders within large-scale Critical Infrastructure Networks (CINs). Privacy and efficiency constraints make full-fledged ICS data audits impractical, creating a key challenge for regulators seeking trustworthy verification of local detection outcomes. The absence of private, verifiable mechanisms can degrade network-wide situational awareness and increase false positives, resulting in unnecessary downtimes.

To overcome these limitations, we present zkSTAR, a zkSNARK-based framework that enables publicly verifiable attack detection while preserving complete zero knowledge of local ICS state-space dynamics. zkSTAR provides formal guarantees of temporal and statistical consistency for utility-level detections, linking proof verification directly to the authenticity of attack alarms. We theoretically prove knowledge-soundness claims in zkSTAR while demonstrating its robustness to attack suppression in case of compromised utilities. For scalability, the framework decomposes proof generation into distinct zkSNARK-based kernels, ensuring computational efficiency and modularity.

We evaluate zkSTAR on a containerized testbed using real-world datasets from HAI and ORNL-PS, demonstrating robust performance across varying TC intervals and precision scales. System-level analysis further confirms the framework’s stability in terms of CPU and memory utilization. Overall, our results show that zkSTAR enables secure, privacy-preserving, and publicly verifiable cyberattack detection suitable for diverse CIN deployment scenarios.

\newpage
\section*{Acknowledgements}
\noindent \textit{Generative AI Usage}: Generative AI tools such as Claude were used for editorial purposes in the preparation of this manuscript (e.g., grammar refinement, stylistic polishing, and clarity improvements). The authors relied on generative AI tools like Claude for assisting with the creation of skeletal code for certain class and function declarations as well as for routine debugging activities for the core programming logic. All AI-assisted content was reviewed, verified, and validated by the authors, including re-running experiments and checking citations.

\bibliographystyle{ieeetr} 
\bibliography{main}

@misc{ezkl2024,
  title        = {ezkl: Zero-Knowledge Machine Learning Inference Framework},
  author       = {South, Tobin and Camuto, Alexander and contributors},
  howpublished = {\url{https://github.com/zkonduit/ezkl}},
  note         = {Accessed: 2025-02-12},
  year         = {2024}
}

@article{langner2011stuxnet,
  author    = {Ralph Langner},
  title     = {Stuxnet: Dissecting a Cyberwarfare Weapon},
  journal   = {IEEE Security \& Privacy},
  volume    = {9},
  number    = {3},
  pages     = {49--51},
  year      = {2011},
  publisher = {IEEE},
  doi       = {10.1109/MSP.2011.67}
}

@inproceedings{salazar2024tale,
  title={A tale of two Industroyers: It was the season of darkness},
  author={Salazar, Luis and Castro, Sebasti{\'a}n R and Lozano, Juan and Koneru, Keerthi and Zambon, Emmanuele and Huang, Bing and Baldick, Ross and Krotofil, Marina and Rojas, Alonso and Cardenas, Alvaro A},
  booktitle={2024 IEEE Symposium on Security and Privacy (SP)},
  pages={312--330},
  year={2024},
  organization={IEEE}
}

@misc{jeffries2022cyber,
  title={Cyber risk to mission case study: Triton},
  author={Jeffries, Blaine and Saravia, Stephanie and Carter, Cedric and Ankuda, Zachary},
  howpublished = {\url{https://apps.dtic.mil/sti/html/trecms/AD1183008/}},
  year={2022},
  organization={Defense Technical Information Center, U.S. Department of Defense}
}

@article{froelicher2020drynx,
  author  = {David Froelicher and Juan R. Troncoso-Pastoriza and Jo{\~a}o S{\'a} Sousa and Jean-Pierre Hubaux},
  title   = {Drynx: Decentralized, Secure, Verifiable System for Statistical Queries and Machine Learning on Distributed Datasets},
  journal = {IEEE Transactions on Information Forensics and Security},
  volume  = {15},
  pages   = {3035--3050},
  year    = {2020},
  doi     = {10.1109/TIFS.2020.2976612}
}

@misc{anonym2025zkSTAR,
  title        = {zkSTAR: Zero-Knowledge Framework for Regulatory Compliance in Industrial Control Systems},
  year         = {2025},
  howpublished = {\url{https://disys-lab.github.io/zkSTAR/}},
  note         = {zkSTAR artifact and source code (anonymous submission for review)},
}

@misc{halo2_zcash_2023,
  title        = {Halo2: The Halo2 Zero-Knowledge Proving System},
  author       = {{Zcash Team}},
  year         = {2023},
  howpublished = {\url{https://github.com/zcash/halo2}},
  note         = {Accessed: 2025-10-26}
}

@article{ramanan2021blockchain,
  title={Blockchain-based decentralized replay attack detection for large-scale power systems},
  author={Ramanan, Paritosh and Li, Dan and Gebraeel, Nagi},
  journal={IEEE Transactions on Systems, Man, and Cybernetics: Systems},
  volume={52},
  number={8},
  pages={4727--4739},
  year={2021},
  publisher={IEEE}
}

@misc{dhs_oig2,
title={Additional Progress Needed to Improve Information Sharing under the Cybersecurity Act of 2015 },
author={Joseph V. Cuffari},
publisher={Office of Inspector General, U.S. Department of Homeland Security},
year={2022},
howpublished = {\url{https://www.oig.dhs.gov/sites/default/files/assets/2022-08/OIG-22-59-Aug22.pdf}}
}

@book{nolan2015cybersecurity,
  title={Cybersecurity and information sharing: Legal challenges and solutions},
  author={Nolan, Andrew},
  volume={5},
  year={2015},
  publisher={Congressional Research Service}
}

@misc{hai,
    author={Shin, Hyeok-Ki and Lee, Woomyo and Choi, Seungoh and Yun, Jeong-Han and Min, Byung-Gi},
    title={HAI security datasets},
    year={2023},
    url={https://github.com/icsdataset/hai},
 }

@inproceedings{coskun2017long,
  title={Long short-term memory kalman filters: Recurrent neural estimators for pose regularization},
  author={Coskun, Huseyin and Achilles, Felix and DiPietro, Robert and Navab, Nassir and Tombari, Federico},
  booktitle={Proceedings of the IEEE International Conference on Computer Vision},
  pages={5524--5532},
  year={2017}
}

@inproceedings{narayan2015verifiable,
  title={Verifiable differential privacy},
  author={Narayan, Arjun and Feldman, Ariel and Papadimitriou, Antonis and Haeberlen, Andreas},
  booktitle={Proceedings of the Tenth European Conference on Computer Systems},
  pages={1--14},
  year={2015}
}

@inproceedings{ben2013snarks,
  title={SNARKs for C: Verifying program executions succinctly and in zero knowledge},
  author={Ben-Sasson, Eli and Chiesa, Alessandro and Genkin, Daniel and Tromer, Eran and Virza, Madars},
  booktitle={Advances in Cryptology--CRYPTO 2013: 33rd Annual Cryptology Conference, Santa Barbara, CA, USA, August 18-22, 2013. Proceedings, Part II},
  pages={90--108},
  year={2013},
  organization={Springer}
}

@inproceedings{costello2015geppetto,
  title={Geppetto: Versatile verifiable computation},
  author={Costello, Craig and Fournet, C{\'e}dric and Howell, Jon and Kohlweiss, Markulf and Kreuter, Benjamin and Naehrig, Michael and Parno, Bryan and Zahur, Samee},
  booktitle={2015 IEEE Symposium on Security and Privacy},
  pages={253--270},
  year={2015},
  organization={IEEE}
}

@inproceedings{huang2022zkmlaas,
  title={zkMLaaS: a Verifiable Scheme for Machine Learning as a Service},
  author={Huang, Chenyu and Wang, Jianzong and Chen, Huangxun and Si, Shijing and Huang, Zhangcheng and Xiao, Jing},
  booktitle={GLOBECOM 2022-2022 IEEE Global Communications Conference},
  pages={5475--5480},
  year={2022},
  organization={IEEE}
}

@article{wang2022ezdps,
  title={ezDPS: An Efficient and Zero-Knowledge Machine Learning Inference Pipeline},
  author={Wang, Haodi and Hoang, Thang},
  journal={arXiv preprint arXiv:2212.05428},
  year={2022}
}

@article{nguyen2022preserving,
  title={Preserving privacy and security in federated learning},
  author={Nguyen, Truc and Thai, My T},
  journal={arXiv preprint arXiv:2202.03402},
  year={2022}
}

@article{yang2023fedzkp,
  title={FedZKP: Federated Model Ownership Verification with Zero-knowledge Proof},
  author={Yang, Wenyuan and Yin, Yuguo and Zhu, Gongxi and Gu, Hanlin and Fan, Lixin and Cao, Xiaochun and Yang, Qiang},
  journal={arXiv preprint arXiv:2305.04507},
  year={2023}
}

@article{li2020detection,
  title={Detection and differentiation of replay attack and equipment faults in SCADA systems},
  author={Li, Dan and Gebraeel, Nagi and Paynabar, Kamran},
  journal={IEEE Transactions on Automation Science and Engineering},
  volume={18},
  number={4},
  pages={1626--1639},
  year={2020},
  publisher={IEEE}
}

@article{li2021degradation,
  title={A degradation-based detection framework against covert cyberattacks on SCADA systems},
  author={Li, Dan and Paynabar, Kamran and Gebraeel, Nagi},
  journal={IISE Transactions},
  volume={53},
  number={7},
  pages={812--829},
  year={2021},
  publisher={Taylor \& Francis}
}

@article{li2022online,
  title={An online approach to covert attack detection and identification in power systems},
  author={Li, Dan and Gebraeel, Nagi and Paynabar, Kamran and Meliopoulos, AP Sakis},
  journal={IEEE Transactions on Power Systems},
  volume={38},
  number={1},
  pages={267--277},
  year={2022},
  publisher={IEEE}
}

@article{hwang2017chi,
  title={Chi-square quantile-based multivariate variance monitoring for individual observations},
  author={Hwang, Wook-Yeon},
  journal={Communications in Statistics-Simulation and Computation},
  volume={46},
  number={7},
  pages={5392--5409},
  year={2017},
  publisher={Taylor \& Francis}
}

@article{li2020deep,
  title={Deep Learning based Covert Attack Identification for Industrial Control Systems},
  author={Li, Dan and Ramanan, Paritosh and Gebraeel, Nagi and Paynabar, Kamran},
  journal={arXiv preprint arXiv:2009.12360},
  year={2020}
}

@inproceedings{drias2015taxonomy,
  title={Taxonomy of attacks on industrial control protocols},
  author={Drias, Zakarya and Serhrouchni, Ahmed and Vogel, Olivier},
  booktitle={2015 International Conference on Protocol Engineering (ICPE) and International Conference on New Technologies of Distributed Systems (NTDS)},
  pages={1--6},
  year={2015},
  organization={IEEE}
}

@article{smith2011decoupled,
  title={A decoupled feedback structure for covertly appropriating networked control systems},
  author={Smith, Roy S},
  journal={IFAC Proceedings Volumes},
  volume={44},
  number={1},
  pages={90--95},
  year={2011},
  publisher={Elsevier}
}

@inproceedings{urbina2016limiting,
  title={Limiting the impact of stealthy attacks on industrial control systems},
  author={Urbina, David I and Giraldo, Jairo A and Cardenas, Alvaro A and Tippenhauer, Nils Ole and Valente, Junia and Faisal, Mustafa and Ruths, Justin and Candell, Richard and Sandberg, Henrik},
  booktitle={Proceedings of the 2016 ACM SIGSAC Conference on Computer and Communications Security},
  pages={1092--1105},
  year={2016},
  organization={ACM}
}

@misc{doe,
title={Energy Sector Cybersecurity Preparedness
},
publisher={Office of Cybersecurity, Energy Security and Emergency Response, U.S. Department of Energy},
howpublished = {\url{https://www.energy.gov/ceser/activities/cybersecurity-critical-energy-infrastructure/energy-sector-cybersecurity}}
}

@inproceedings{kang2014cyber,
  title={Cyber threats and defence approaches in SCADA systems},
  author={Kang, Dong-Ho and Kim, Byoung-Koo and Na, Jung-Chan},
  booktitle={16th International Conference on Advanced Communication Technology},
  pages={324--327},
  year={2014},
  organization={IEEE}
}

@article{huang2018online,
  title={An Online Detection Framework for Cyber Attacks on Automatic Generation Control},
  author={Huang, Tong and Satchidanandan, Bharadwaj and Kumar, PR and Xie, Le},
  journal={IEEE Transactions on Power Systems},
  volume={33},
  number={6},
  pages={6816--6827},
  year={2018},
  publisher={IEEE}
}

@inproceedings{hoehn2016detection,
  title={Detection of covert attacks and zero dynamics attacks in cyber-physical systems},
  author={Hoehn, Andreas and Zhang, Ping},
  booktitle={American Control Conference (ACC), 2016},
  pages={302--307},
  year={2016},
  organization={IEEE}
}

@article{pan2015classification,
  title={Classification of disturbances and cyber-attacks in power systems using heterogeneous time-synchronized data},
  author={Pan, Shengyi and Morris, Thomas and Adhikari, Uttam},
  journal={IEEE Transactions on Industrial Informatics},
  volume={11},
  number={3},
  pages={650--662},
  year={2015},
  publisher={IEEE}
}

@article{pan2015developing,
  title={Developing a hybrid intrusion detection system using data mining for power systems},
  author={Pan, Shengyi and Morris, Thomas and Adhikari, Uttam},
  journal={IEEE Transactions on Smart Grid},
  volume={6},
  number={6},
  pages={3104--3113},
  year={2015},
  publisher={IEEE}
}

@article{rahman2017multi,
  title={Multi-agent approach for enhancing security of protection schemes in cyber-physical energy systems},
  author={Rahman, Md Shihanur and Mahmud, Md Apel and Oo, Aman Maung Than and Pota, Hemanshu Roy},
  journal={IEEE Transactions on Industrial Informatics},
  volume={13},
  number={2},
  pages={436--447},
  year={2017},
  publisher={IEEE}
}

@inproceedings{caselli2016specification,
  title={Specification Mining for Intrusion Detection in Networked Control Systems.},
  author={Caselli, Marco and Zambon, Emmanuele and Amann, Johanna and Sommer, Robin and Kargl, Frank},
  booktitle={USENIX Security Symposium},
  pages={791--806},
  year={2016}
}

@article{ye2004robustness,
  title={Robustness of the Markov-chain model for cyber-attack detection},
  author={Ye, Nong and Zhang, Yebin and Borror, Connie M},
  journal={IEEE Transactions on Reliability},
  volume={53},
  number={1},
  pages={116--123},
  year={2004},
  publisher={IEEE}
}

@inproceedings{teixeira2010cyber,
  title={Cyber security analysis of state estimators in electric power systems},
  author={Teixeira, Andr{\'e} and Amin, Saurabh and Sandberg, Henrik and Johansson, Karl Henrik and Sastry, Shankar S},
  booktitle={49th IEEE Conference on Decision and Control (CDC). Atlanta, GA. DEC 15-17, 2010},
  pages={5991--5998},
  year={2010}
}

@article{mo2014detecting,
  title={Detecting integrity attacks on SCADA systems},
  author={Mo, Yilin and Chabukswar, Rohan and Sinopoli, Bruno},
  journal={IEEE Transactions on Control Systems Technology},
  volume={22},
  number={4},
  pages={1396--1407},
  year={2014},
  publisher={IEEE}
}

@incollection{goldwasser2019knowledge,
  title={The knowledge complexity of interactive proof-systems},
  author={Goldwasser, Shafi and Micali, Silvio and Rackoff, Chales},
  booktitle={Providing Sound Foundations for Cryptography: On the Work of Shafi Goldwasser and Silvio Micali},
  pages={203--225},
  year={2019}
}

@incollection{blum2019non,
  title={Non-interactive zero-knowledge and its applications},
  author={Blum, Manuel and Feldman, Paul and Micali, Silvio},
  booktitle={Providing Sound Foundations for Cryptography: On the Work of Shafi Goldwasser and Silvio Micali},
  pages={329--349},
  year={2019}
}

@inproceedings{kilian1992note,
  title={A note on efficient zero-knowledge proofs and arguments},
  author={Kilian, Joe},
  booktitle={Proceedings of the twenty-fourth annual ACM symposium on Theory of computing},
  pages={723--732},
  year={1992}
}

@article{micali2000computationally,
  title={Computationally sound proofs},
  author={Micali, Silvio},
  journal={SIAM Journal on Computing},
  volume={30},
  number={4},
  pages={1253--1298},
  year={2000},
  publisher={SIAM}
}

@inproceedings{gennaro2013quadratic,
  title={Quadratic span programs and succinct NIZKs without PCPs},
  author={Gennaro, Rosario and Gentry, Craig and Parno, Bryan and Raykova, Mariana},
  booktitle={Advances in Cryptology--EUROCRYPT 2013: 32nd Annual International Conference on the Theory and Applications of Cryptographic Techniques, Athens, Greece, May 26-30, 2013. Proceedings 32},
  pages={626--645},
  year={2013},
  organization={Springer}
}

@article{parno2016pinocchio,
  title={Pinocchio: Nearly practical verifiable computation},
  author={Parno, Bryan and Howell, Jon and Gentry, Craig and Raykova, Mariana},
  journal={Communications of the ACM},
  volume={59},
  number={2},
  pages={103--112},

}

@inproceedings{ben2014succinct,
  title={Succinct $\{$Non-Interactive$\}$ zero knowledge for a von neumann architecture},
  author={Ben-Sasson, Eli and Chiesa, Alessandro and Tromer, Eran and Virza, Madars},
  booktitle={23rd USENIX Security Symposium (USENIX Security 14)},
  pages={781--796},
  year={2014}
}



\appendices

\section{Extended Kalman Filter Model and Training}
\subsection{EKF Modeling for ICS}\label{app:ekf_app}
Our state space modeling framework is based on the extended Kalman filter (EKF) modeling paradigm that is capable of handling non-linearities inherent in the state space process \cite{coskun2017long}. In order to represent the EKF model, we use a deep neural network driven strategy that consists of one feed-forward and three LSTM sub-modules. Specifically, we consider a sensor-driven non-linear system at time $t$, where $x_t\in R^{m}$ represents the latent space embedding, {$u_t\in R^{m}$} represents the control action and $y_t\in R^{d}$ represents noisy sensor measurements from asset sensors.
\begin{gather}
    x_{t+1}= g(x_{t-1},u_{t-1}) + v_t,\label{eq:eqn1a}\\
    y_t=h(x_t)+w_t, \label{eq:eqn1b}
\end{gather}
In Equations \eqref{eq:eqn1a}, \eqref{eq:eqn1b}, {${g}, {h}$} are the state transition and observation functions respectively. The process and measurement noises at time $t$ are denoted by $v_t\in R^{m}$,  and $w_t\in R^{m}$ respectively. The process and measurement noises follow multivariate normal distributions with zero mean implying that $v_t \sim N(0,Q_t), w_t \sim N(0,R_t)$, where $Q_t,R_t$ represent the covariance matrices respectively. Such a type of modeling framework has also been used extensively in prior art \cite{smith2011decoupled, mo2014detecting, teixeira2010cyber}. Figure \ref{fig:temp_consist_core} represents the dynamics of our LSTM based deep neural network that is specifically tailored to establish the extended Kalman filter model for ICS subsystems.

\begin{figure}[!htb]
    \centering
    \includegraphics[width=0.48\textwidth,keepaspectratio]{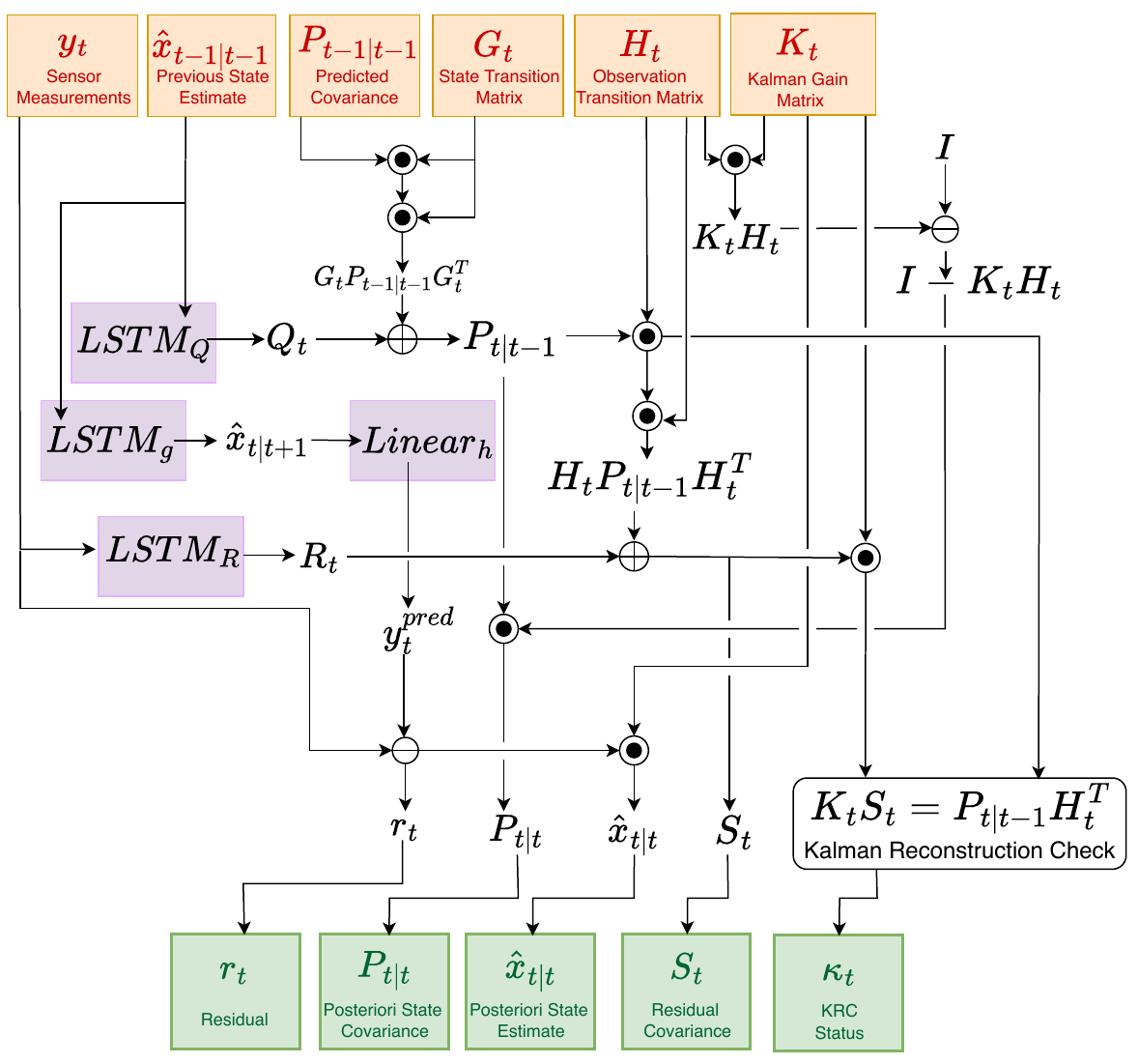}
  \caption{Non Linear Kalman Filter based Temporal State Estimation}\label{fig:temp_consist_core}
\end{figure}

\subsection{EKF training methodology}\label{subsec:ekf_train}
During the training phase, the efficacy of the model is assessed by comparing the estimated state $ x_{t|t} $ against the ground-truth state $x_t$ through a carefully constructed loss function. This function typically comprises a mean squared error (MSE) term to minimize the divergence between the estimated and true states. To ensure statistical robustness and numerical stability, regularization terms are incorporated to constrain the norms and conditioning of the covariance matrices $ Q_t$ and $R_t$, as well as the Kalman gain $ K_t $. These constraints mitigate overfitting and foster consistent learning behavior across iterations.

The total loss is subsequently backpropagated through the model's architecture via automatic differentiation techniques, enabling the optimization of all learnable parameters. This includes parameters embedded within the state transition function, the measurement function, and the covariance estimation modules. Through iterative training over sequential data, the KFLSTM model progressively refines its ability to encapsulate the system's dynamic behavior and uncertainty propagation, aligning seamlessly with the theoretical underpinnings of non-linear Kalman filtering.

\section{Theorem and Lemma Proofs}
\subsection{Proof of Theorem \ref{thm:TC}}\label{app:TC_proof}
\begin{proof}
We know from Definition \ref{defn1}, for practical purposes $\mathcal{H}^{out}_t$ and $\mathcal{H}^{in}_{t+1}$ can be used to uniquely identify $x^{out}_t$ and $x^{in}_{t+1}$. Therefore, it is sufficient to demonstrate the following
\begin{gather}
\Phi_{t+1} = 1  \implies \mathcal{H}^{in}_{t+1} = \mathcal{H}^{out}_{t}\\
\mathcal{H}^{in}_{t+1} = \mathcal{H}^{out}_{t} \implies \Phi_{t+1} = 1
\end{gather}
(Case 1 $\Phi_{t+1} = 1  \implies \mathcal{H}^{in}_{t+1} = \mathcal{H}^{out}_{t}$): Assume that for some \(t,t+1\), there exists a $\tilde{x}^{in}_{t+1} \neq x^{out}_t$ such that $\Phi_{t+1} = 1$. Consequently, we can establish $\tilde{\Delta}^{in}_{t+1} = (\tilde{\mathcal{H}}^{in}_{t+1},\tilde{x}^{in}_{t+1},y_{t+1})$, where $\tilde{\mathcal{H}}^{in}_{t} = \verb|Comm|(\tilde{x}^{in}_{t+1})$. Using Definitions \ref{defn4} and \ref{defn5}, we can obtain $\tilde{\Pi}_{t+1},\tilde{\Phi}_{t+1}$ such that the following equation holds.
\begin{gather}
\tilde{\Pi}_{t+1} = \texttt{Prove}(pk,\theta,\mathcal{M},\tilde{\Delta}^{in}_{t+1},\tilde{\Delta}^{out}_{t+1}) \label{eq:41c11a}\\
\tilde{\Phi}_{t+1} = \texttt{Verify}(vk,\tilde{\Pi}_{t+1},\tilde{\Delta}^{in}_{t+1},\tilde{\Delta}^{out}_{t+1}) = 1 \label{eq:41c11}
\end{gather}
Using $\mathcal{H}^{out}_t$ to construct $\Delta^{in}_{t+1} = (\mathcal{H}^{out}_t,x^{out}_{t+1},y_{t+1})$ gives
\begin{gather}
\Phi_{t+1} = \texttt{Verify}(vk,\tilde{\Pi}_{t+1},\Delta^{in}_{t+1},\tilde{\Delta}^{out}_{t+1})\label{eq:41c12}
\end{gather}
By the soundness argument of the zkSNARK proving system, we know that $\Phi_{t+1} \neq \tilde{\Phi}_{t+1}$ since verification statements in Equations \eqref{eq:41c11} and \eqref{eq:41c12} cannot hold simultaneously. More specifically, proof artifacts generated for public inputs corresponding to $\tilde{\mathcal{H}}^{in}_{t+1}$ cannot verify against a different public input $\mathcal{H}^{out}_t \neq \tilde{\mathcal{H}}^{in}_{t+1}$.  As a result $\Phi_{t+1} \neq 1$ which presents a contradiction to our stated assumptions.\\
\noindent(Case 2 $\mathcal{H}^{in}_{t+1} = \mathcal{H}^{out}_{t} \implies \Phi_{t+1} = 1$): Assume that for some $t+1$, $\Phi_{t+1} \neq 1$ such that $\mathcal{H}^{in}_{t+1} = \mathcal{H}^{out}_{t}$. Further let $\tilde{\Delta}^{in}_{t+1} = (\mathcal{H}^{out}_t,x^{in}_{t+1},y_{t+1})$ and $\Delta^{in}_{t+1} = (\mathcal{H}^{in}_{t+1},x^{out}_{t},y_{t+1})$. Using Definition \ref{defn2}, we can assert that $\Delta^{in}_{t+1}$ and $\tilde{\Delta}^{in}_{t+1}$ are both valid inputs to $\mathcal{M}$ parametrized by $\theta$. Additionally, we can also assert that there exist $\Pi_{t+1}, \Delta^{out}_{t+1}, \tilde{\Pi}_{t+1},  \tilde{\Delta}^{out}_{t+1}$ such that
\begin{gather}
\Phi_{t+1} = \texttt{Verify}(vk,\Pi_{t+1},\Delta^{in}_{t+1},\Delta^{out}_{t+1}) = 0 \label{eq:41c21}\\
\tilde{\Phi}_{t+1} = \texttt{Verify}(vk,\tilde{\Pi}_{t+1},\tilde{\Delta}^{in}_{t+1},\tilde{\Delta}^{out}_{t+1}) = 1\label{eq:41c22}
\end{gather}
Since we assume that $\mathcal{H}^{in}_{t+1} = \mathcal{H}^{out}_{t}$, we know that $\tilde{\Delta}^{in}_{t+1} = \Delta^{in}_{t+1}$ due to commitment binding. 
This means that for $\Phi_{t+1} = 1$ either the proof artifact $\Pi_{t+1}$ is invalid or $\Delta^{out}_{t+1}$ is the incorrect output corresponding to $\Delta^{in}_{t+1}$. If the proof is incorrect, it means that $\Phi_{t+1}=0$ holds trivially which invalidates $\Delta^{in}_{t+1},\Delta^{out}_{t+1}$ as viable inputs and outputs to the anomaly detector contradicting our assumption of hash consistency. 
On the other hand, since $\mathcal{M}(\cdot,\theta)$ is deterministic, given a valid proof $\tilde{\Pi}_{t+1}$ corresponding to an input tuple $\Delta^{in}_{t+1}=\tilde\Delta^{in}_{t+1}$ must necessarily result in $\Delta^{out}_{t+1} = \tilde\Delta^{out}_{t+1}$. This implies that $\Phi_{t+1} = 1$ which directly contradicts our assumption.
\end{proof}

\subsection{Proof of Lemma \ref{lem:LM1}}\label{subsec:LM1_proof}

\begin{proof}
Suppose $\tilde{y}_t$ and $y_t$ are probabilistically similar per Definition~\ref{defn8}. By Definition~\ref{defn8}, either (i) $\tilde{T}_t > \chi^2_{p,\alpha}$ and $T_t > \chi^2_{p,\alpha}$, or (ii) $\tilde{T}_t \leq \chi^2_{p,\alpha}$ and $T_t \leq \chi^2_{p,\alpha}$. In case~(i), both measurements trigger an alarm, so $\tilde{\rho}_t = \rho_t = 1$. In case~(ii), neither measurement triggers an alarm, so $\tilde{\rho}_t = \rho_t = 0$. In both cases $\tilde{\rho}_t = \rho_t$, establishing statistical consistency.
\end{proof}

\subsection{Proof of Theorem \ref{thm:SC}}\label{subsec:SC_proof}
\begin{proof}
We know from Theorem \ref{thm:TC} that $\mathcal{H}^{in}_{t+1} = \mathcal{H}^{out}_{t} \text{ , and } \Phi_{t+1} = 1$ is a necessary and sufficient condition for temporal consistency for any given measurement $y_{t+1}$. We also note that the temporal consistency check $\mathcal{H}^{in}_{t+1} = \mathcal{H}^{out}_{t}$ is independent of the underlying $y_{t+1}$.
For proving this theorem, we assume that $\rho^{real}_{t+1}$ and $\rho^{rep}_{t+1}$ represent the alarm outcomes pertaining to the true post attack and the reported sensor data values $y^{real}_{t+1}$ and $y^{rep}_{t+1}$ respectively. Therefore, it suffices to show that for a temporally consistent system with $y^{rep}_{t+1}$, the verification status $\Phi^{rep}_{t+1} = 1$ is a necessary and sufficient condition for statistical consistency to hold with respect to $y^{real}_{t+1}$ subject to EKF sensitivity and Type 1 error probability bounds.\\ 
\noindent(Case 1: $\Phi^{rep}_{t+1} = 1 \implies$ statistical 
consistency with probability $\geq 1 - \epsilon_{max} - \alpha$): Assume $\Phi^{rep}_{t+1} = 1$. There exists a system of proof artifacts pertaining to $\rho^{rep}_{t+1}$ satisfying
\begin{gather}
\tilde{\Delta}^{in}_{t+1} = (\mathcal{H}^{in}_{t+1},x^{in}_{t+1},y^{rep}_{t+1})\\
\tilde{\Delta}^{out}_{t+1} = (\mathcal{H}^{out}_{t+1},\rho^{rep}_{t+1},{x}^{out}_{t+1}, r_{t+1}, K_{t+1}, Q_{t+1}, R_{t+1})\\
\tilde{\Pi}_{t+1} = \texttt{Prove}(pk,\theta,\mathcal{M},\tilde{\Delta}^{in}_{t+1},\tilde{\Delta}^{out}_{t+1})\\
\Phi^{rep}_{t+1} = \texttt{Verify}(vk,\tilde{\Pi}_{t+1},\tilde{\Delta}^{in}_{t+1},\tilde{\Delta}^{out}_{t+1}) = 1
\end{gather}
Since $\mathcal{M}(\cdot;\theta)$ is deterministic and $\Phi^{rep}_{t+1}=1$, the reported alarm $\rho^{rep}_{t+1}$ is a committed deterministic outcome of the EKF executed on $y^{rep}_{t+1}$. Under this circumstance, we consider two types of possible events involving $y^{real}_{t+1}$ and $y^{rep}_{t+1}$.

We note that $y^{real}_{t+1}$ being probabilistically similar to $y^{rep}_{t+1}$ guarantees $\rho^{rep}_{t+1} = \rho^{real}_{t+1}$ under Lemma \ref{lem:LM1}. On the other hand, $y^{real}_{t+1}$ being probabilistically dissimilar to $y^{rep}_{t+1}$ is the only event type under which $\rho^{rep}_{t+1} \neq \rho^{real}_{t+1}$ is possible. The event $\rho^{rep}_{t+1} \neq \rho^{real}_{t+1}$ corresponds to one of two cases pertaining to either $T^{rep}_{t+1} < \chi^2_{m,\alpha}$ given $T^{real}_{t+1} > \chi^2_{m,\alpha}$, or $T^{rep}_{t+1} > \chi^2_{m,\alpha}$ given $T^{real}_{t+1} < \chi^2_{m,\alpha}$. By Assumption~\ref{as:as3}, the probability of $T^{rep}_{t+1} < \chi^2_{m,\alpha}$ given $T^{real}_{t+1} > \chi^2_{m,\alpha}$ is bounded by $\epsilon_{max}$. On the other hand, $T^{rep}_{t+1} > \chi^2_{m,\alpha}$ given $T^{real}_{t+1} < \chi^2_{m,\alpha}$ corresponds to a Type 1 error of the $\chi^2$ anomaly detector which is bounded by Type 1 error rate $\alpha$. Therefore,
\[
\Pr\Big[\rho^{rep}_{t+1} \neq \rho^{real}_{t+1}\Big] \leq \epsilon_{max} + \alpha
\]
and statistical consistency holds with probability at least $1 - \epsilon_{max} - \alpha$.

\noindent(Case 2: statistical consistency $\implies 
\Phi^{rep}_{t+1} = 1$): Assume $\Phi^{rep}_{t+1} \neq 1$ and $\rho^{real}_{t+1} = 
\rho^{rep}_{t+1}$ with $\Phi^{real}_{t+1}$ denoting the verification outcome for the true post attack measurement. As a result, we consider a proof artifact $\Pi_{t+1}$ that encodes the mapping of a legal input ${\Delta}^{in}_{t+1}$ to a legal output ${\Delta}^{out}_{t+1}$ such that the following holds:
\begin{gather}
{\Delta}^{in}_{t+1} = (\mathcal{H}^{in}_{t+1},x^{in}_{t+1},y^{rep}_{t+1})\\
{\Delta}^{out}_{t+1}=(\mathcal{H}^{out}_{t+1},\rho^{rep}_{t+1},{x}^{out}_{t+1}, r^{rep}_{t+1}, K_{t+1}, Q_{t+1}, R_{t+1})\\
\Phi^{rep}_{t+1} = \texttt{Verify}(vk,{\Pi}_{t+1},{\Delta}^{in}_{t+1},{\Delta}^{out}_{t+1}) = 0
\end{gather}
Alternately, we also know that there exists a proof artifact $\tilde \Pi_{t+1}$ that maps a legal input $\tilde{\Delta}^{in}_{t+1}$ to a legal output $\tilde{\Delta}^{out}_{t+1}$ such that the following holds:
\begin{gather}
\tilde{\Delta}^{in}_{t+1} = (\tilde{\mathcal{H}}^{in}_{t+1},\tilde{x}^{in}_{t+1},y^{real}_{t+1})\\
\tilde{\Delta}^{out}_{t+1} = (\mathcal{H}^{out}_{t+1},\rho^{real}_{t+1},\tilde {x}^{out}_{t+1},\tilde r_{t+1},\tilde K_{t+1},\tilde Q_{t+1},\tilde R_{t+1})\\{\Phi}^{real}_{t+1} = \texttt{Verify}(vk,\tilde{\Pi}_{t+1},\tilde{\Delta}^{in}_{t+1},\tilde{\Delta}^{out}_{t+1}) = 1
\end{gather}
Given $\tilde\Delta^{out}_{t+1}, \Delta^{out}_{t+1}$ are deterministic outcomes of $\mathcal{M}(\cdot, \theta)$, it suffices to only analyze the condition surrounding $\tilde\Delta^{in}_{t+1}, \Delta^{in}_{t+1}$. Moreover, we know that $\Phi^{rep}_{t+1} = 0$ can only occur if the proof artifact $\Pi_{t+1}$ fails to map ${\Delta}^{in}_{t+1}$ to ${\Delta}^{out}_{t+1}$. This could be either due to $\Pi_{t+1}$ being the incorrect proof artifact for ${\Delta}^{in}_{t+1}, {\Delta}^{out}_{t+1}$, or ${\Delta}^{in}_{t+1}$ being the incorrect input to the correct $\Pi_{t+1}$. 
An incorrect proof trivially yields $\Phi^{rep}_{t+1}=0$ which means that $\rho^{rep}_{t+1}$ is invalidated as a potential outcome of the anomaly detector which contradicts our assumption of SC. Therefore, for statistical consistency to hold, $\Phi^{real}_{t+1} = 1$, $\Phi^{rep}_{t+1} = 0$ can occur only if  $\tilde\Delta^{in}_{t+1} \neq \Delta^{in}_{t+1}$. 

As a result, we are left with two possibilities that justify $\tilde\Delta^{in}_{t+1} \neq \Delta^{in}_{t+1}$, either $y^{real}_{t+1} \neq y^{rep}_{t+1}$ or $(\mathcal{H}^{in}_{t+1}, x^{in}_{t+1}) \neq (\tilde{\mathcal{H}}^{in}_{t+1}, \tilde{x}^{in}_{t+1})$. Given that our zkSNARK circuits are predetermined and fixed, $y^{real}_{t+1}$ and $y^{rep}_{t+1}$ must necessarily generate proof artifacts that successfully verify, creating a contradiction to our supposition of $\Phi^{rep}_{t+1} = 0$.

Consequently, we are left with $(\mathcal{H}^{in}_{t+1}, x^{in}_{t+1}) \neq (\tilde{\mathcal{H}}^{in}_{t+1}, \tilde{x}^{in}_{t+1})$. We know that $\Phi^{real}_{t+1} = 1$, guarantees that $(\tilde{\mathcal{H}}^{in}_{t+1}, \tilde{x}^{in}_{t+1})$ is temporally consistent with previous timestep $t$ due to Theorem \ref{thm:TC}. Therefore, $(\mathcal{H}^{in}_{t+1}, x^{in}_{t+1}) \neq (\tilde{\mathcal{H}}^{in}_{t+1}, \tilde{x}^{in}_{t+1})$ indicates that $(\mathcal{H}^{in}_{t+1}, x^{in}_{t+1})$ lacks temporal consistency. As a result, the claim of $\rho^{rep}_{t+1}$ as the outcome of the anomaly detector is invalid which contradicts our assumption of statistical consistency. 

As a result, $\Phi^{rep}=1$ holds under conditions of statistical consistency which directly substantiates temporal consistency claims.

\end{proof}

\subsection{Proof of Theorem \ref{thm:dsks}}\label{subsec:proof_dsks}
\begin{proof}
We consider the case wherein $\mathcal{A}$ produces $\Pi^{rep}_t$ such that $\Phi^{rep}_t = 1$ while $\Psi^{rep}_t=1$ for any time $t>t^{attack}$. 
By knowledge soundness of the Halo2 proving system \cite{halo2_zcash_2023}, if $\Phi^{rep}_t = 1$ implies that there exists an efficient extractor that can recover the valid witness $\Omega^{rep}_t$ from utility $\mathcal{U}$ with probability at least $1 - \texttt{negl}(\lambda)$. We also know from Theorem \ref{thm:TC} that $\Phi^{rep}_t=1$ guarantees temporal consistency for the underlying system operated on $y^{rep}_t$. Using Theorem \ref{thm:SC}, we can say that statistical consistency also holds with respect to the true post attack measurement $y^{real}_t$ with lower bound $1-\epsilon_{max} -\alpha$. Therefore, for any $t > t^{attack}$ the probability of the adversary to generate a witness ${\Omega}^{rep}_t$ that leads to $\Psi^{rep}_t=1$ while simultaneously ensuring that its supplied proof artifact verifies correctly (i.e. $\Phi^{rep}_t=1$) is given by
\[
    \Pr[\Psi^{rep}_t=1|\Phi^{rep}_t=1] \leq  \epsilon_{max} +\alpha +\texttt{negl}(\lambda) 
\]
\end{proof}

\section{Additional Experimental Results}
\subsection{Proof Setup and Generation Times}\label{subsec:pgt_app_subsec}
\noindent The speedup relation used in Tables \ref{tab:p1_pgt}, \ref{tab:ps_pgt} can be given as:
\begin{gather}\label{eq:speedupcalc}
speedup = D.\frac{\text{(average proving time for D=1)}}{\text{(average proving time for D $\neq$ 1)}}
\end{gather}
\begin{table}[ht]
\centering
\scriptsize
\setlength{\tabcolsep}{3pt}
\renewcommand{\arraystretch}{0.85}
\begin{tabular}{|c|c|c|c|}
\hline
\textbf{Dataset} & \textbf{DVAL} & \textbf{CPU (\%)} & \textbf{Memory (GB)} \\
\hline
HAI     & 16 & 1081.584(604.146) & 30.135(9.980) \\
ORNL-PS & 16 & 1081.584(604.146) & 30.135(9.980) \\
\hline
\end{tabular}
\caption{CPU and Memory Utilization (mean(std)) for Statistical consistency }\label{tab:sysper_sc}
\end{table}
\begin{figure}[!htb]
    \centering
{\includegraphics[width=0.49\textwidth,keepaspectratio]{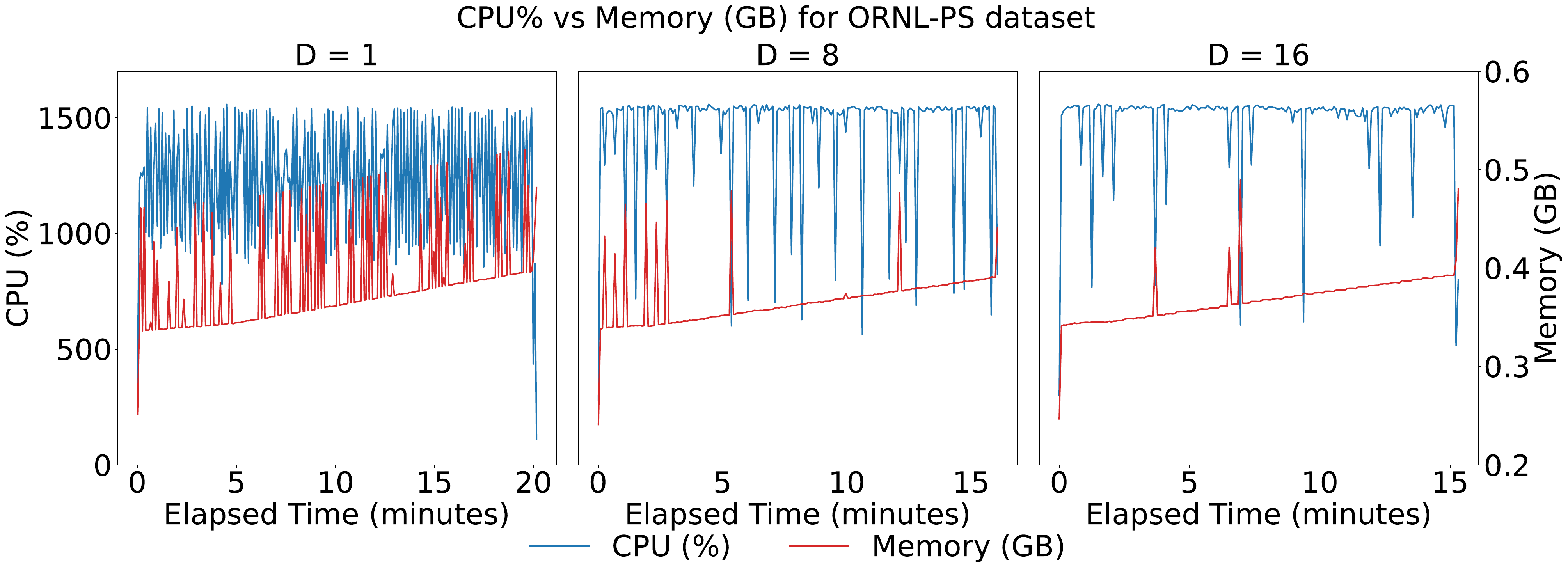}}\label{fig:sysperf_ornl-ps_wg}
  \caption{Witness Generation Performance with varying D for ORNL-PS dataset}\label{fig:ornl-ps_wgp}
\end{figure}

\begin{figure}[!htb]
    \centering
{\includegraphics[width=0.49\textwidth,keepaspectratio]{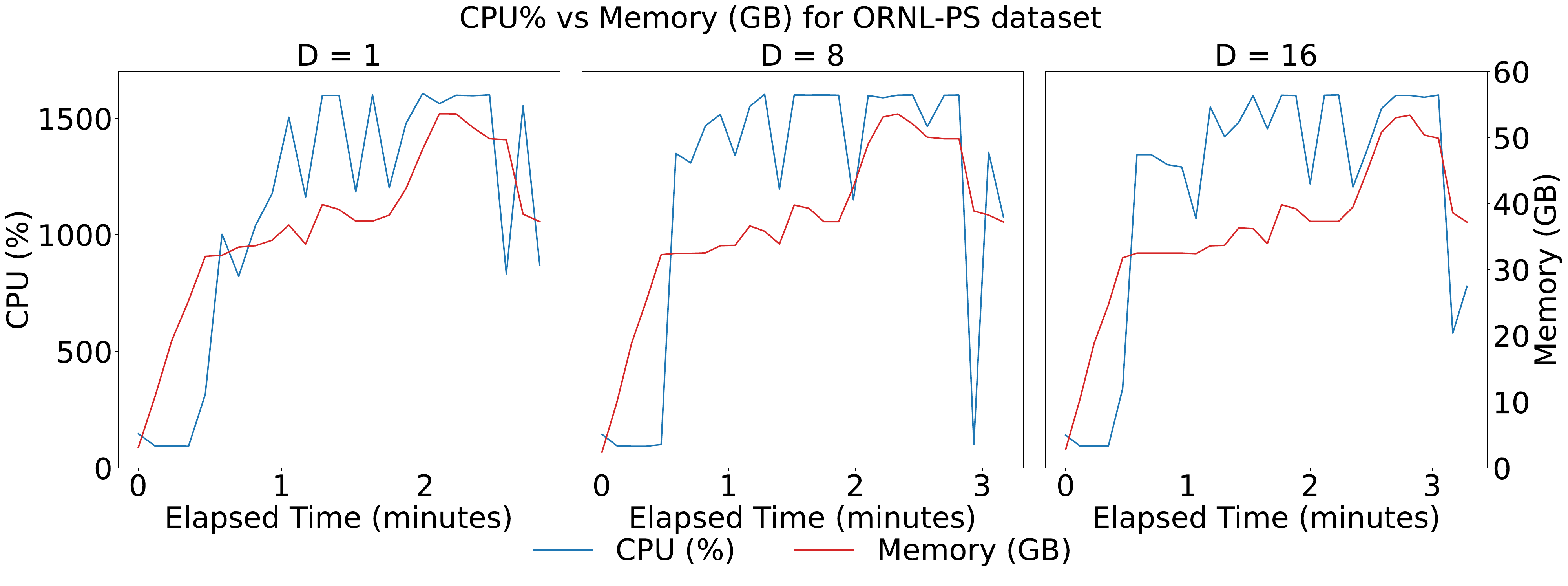}}\label{fig:sysperf_ornl-ps_pp}
  \caption{Prover Performance for Temporal Consistency  with varying D for ORNL-PS dataset}\label{fig:ornl-ps_pptc}
\end{figure}
The speedup computation is based on obtaining the ratio between the average time consumed for the entire set of D unrolled loop iterations when $D=1$ and the corresponding values when $D\neq1$. Since the time horizon appears in the numerator and denominator, they cancel each other resulting in Equation \eqref{eq:speedupcalc}.
\begin{figure*}[!htb]
    \centering
        \includegraphics[width=0.9\textwidth,keepaspectratio]{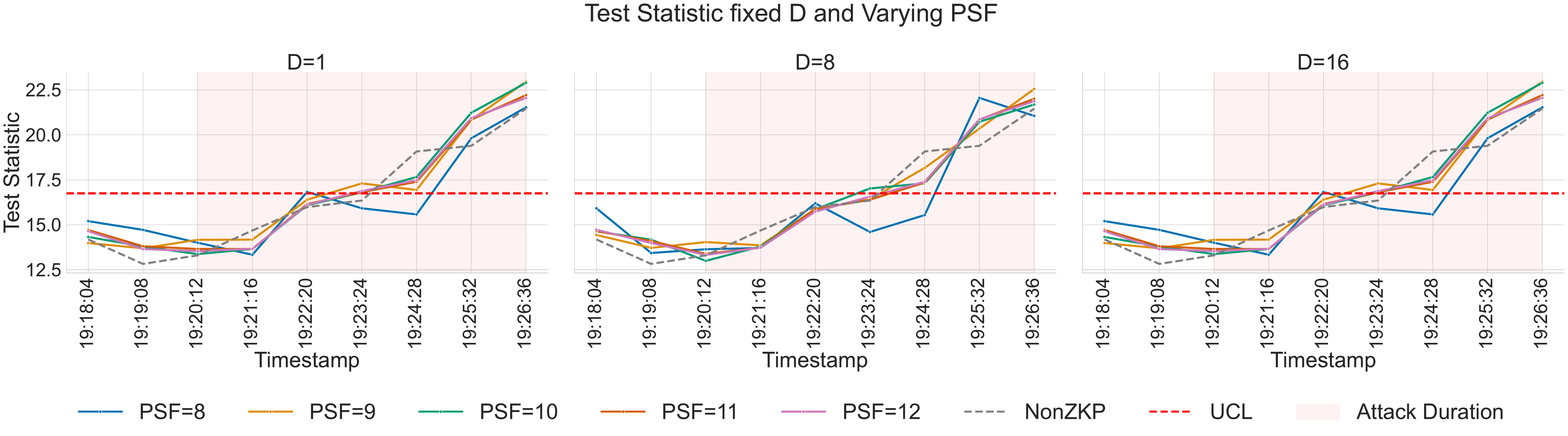}
  \caption{HAI Dataset: Detection Quality with varying PSF and fixed D}\label{fig:P1_scale}
\end{figure*}
\begin{figure*}[!htb]
    \centering
        \includegraphics[width=0.9\textwidth,keepaspectratio]{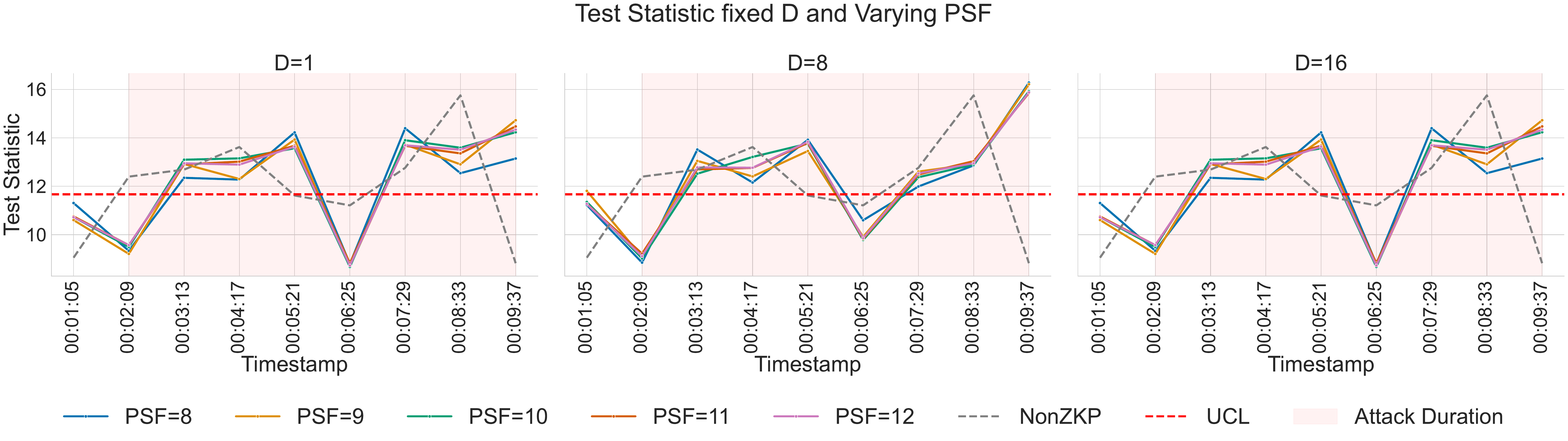}
  \caption{ORNL-PS Dataset: Detection Quality with varying PSF and fixed D}\label{fig:ps_scale}
\end{figure*}
\subsection{Detection Quality with varying PSF}\label{subsec:dq_psf}
In Figures \ref{fig:P1_scale} and \ref{fig:ps_scale}, we depict the detection quality with respect to varying PSF values for distinct values of D. In this case we can observe that the lower PSF values usually exhibit greater volatility in detection, while higher PSF values correspond to a less volatile nature of the test statistic. Additionally, we also observe a slight increase in volatility with increasing D which is expected since the errors in residual computation compound for lower PSF values with higher D values. From these experiments, we can conclude that a higher PSF value might be more desirable for a lower D. However, a mid-range PSF value might also suffice for a higher D value. Regardless, the experiments establish the robustness in detection quality of the zkSTAR framework. 

\begin{table}[t]
\centering
\begin{tabular}{|l|c|c|c|}
\hline
\textbf{Dataset} & \textbf{PK} & \textbf{VK} & \textbf{PF} \\
\hline
HAI     & 16011.879 (2271.245) & 3.878 (1.014) &  0.039 (0.004) \\
ORNL-PS & 16516.104 (2735.173) & 4.103 (1.221) &  0.040 (0.006) \\
\hline
\end{tabular}
\caption{TC Circuit: Proof (PF), Proving Key (PK), and Verification Key (VK) sizes for TC (mean(std) in MB) }\label{tab:ppkv}
\end{table}

\begin{table}[t]
\centering
\begin{tabular}{|l|c|c|c|c|c|}
\hline
\textbf{Dataset} & \textbf{PK} & \textbf{VK} & \textbf{witness} & \textbf{Circuit Size} & \textbf{PF} \\
\hline
HAI & 12931 & 3.002 & 0.023 & 0.009 & 0.037 \\
ORNL-PS     & 12931 & 3.002 & 0.016 & 0.008 & 0.035 \\
\hline
\end{tabular}
\caption{SC Circuit Proof (PF), Proving Key (PK), and Verification Key (VK) sizes for SC in MB) }\label{tab:scppkvws}
\end{table}

\begin{table}[t]
\centering
\begin{tabular}{|c|c|c|c|c|}
\hline
\textbf{D} & \multicolumn{2}{c|}{\textbf{HAI} } & \multicolumn{2}{c|}{\textbf{ORNL-PS}} \\
\cline{2-5}
\textbf{Value}& Witness & Circuit & Witness & Circuit\\
\hline
1  & 0.043(0.001) & 0.25 & 0.036(0.002) & 0.24\\
4  & 0.073(0.001) & 0.78 & 0.064(0.001) & 0.78\\
8  & 0.114(0.001) & 1.50 & 0.101(0.001) & 1.49\\
16 & 0.197(0.001) & 2.94 & 0.176(0.002) & 2.93\\
32 & 0.362(0.002) & 5.82 & 0.325(0.004) & 5.80\\
\hline
\end{tabular}
\caption{Temporal Consistency Witness Size with Varying D (mean(std) in MB)}\label{tab:tcws}
\end{table}

\subsection{Analysis of Proof and Key Sizes} \label{proof_key_sizes}
Tables \ref{tab:ppkv} and \ref{tab:scppkvws} presents the sizes of the proof, proving and verification keys for the temporal consistency and statistical consistency proofs. Table \ref{tab:scppkvws} also presents the witness size pertaining to the statistical consistency proofs. Additionally, we present the witness and the compiled circuit size for temporal consistency proofs for PSF value 8 separately in Table \ref{tab:tcws}. 

In Table \ref{tab:ppkv}, we can observe that for the case of temporal consistency, the proof size as well as proving and verification key sizes are similar for both datasets. On the other hand, in Table \ref{tab:scppkvws}, the circuit, proving and verification key size for the statistical consistency proofs were constant as expected and minor deviations are observed for the proof and witness sizes across both datasets. However, we observe an interesting trend with respect to witness and circuit sizes as represented in Table \ref{tab:tcws}. Mean witness sizes increase gradually from around 43kB to 360kB for HAI and 36kB to 325kB for ORNL-PS. On the other hand, the circuit sizes increase from a smaller value of 0.25MB at D value of 1 to 5.8MB for HAI while exhibiting a near identical trend for ORNL-PS dataset as well. This trend is natural and expected since increase in D value corresponds to larger TC interval sizes making the TC circuit grow exponentially. 

\subsection{Analysis of System Performance of ORNL-PS}\label{subsec:ps_sysper_subsec}
Figures \ref{fig:ornl-ps_wgp} and \ref{fig:ornl-ps_pptc} present the system performance characterized by CPU utilization and memory consumption for ORNL-PS dataset for witness generation and proof generation mechanism. In Table \ref{tab:sysper_sc}, we provide the CPU utilization and memory consumption for the statistical consistency proofs for HAI and ORNL-PS datasets.

\subsection{Analyses of Circuit Setup Costs}\label{subsec:cir_setup_cost}
\begin{table}
\centering
\scriptsize
\setlength{\tabcolsep}{3pt}
\renewcommand{\arraystretch}{0.85}
\begin{tabular}{|l|c|c|c|c|}
\hline
\textbf{Setup} &  \multicolumn{4}{c|}{\textbf{Setup Time (s)}} \\
\cline{2-5}
\textbf{Phase} & \textbf{D=1} & \textbf{D=4} & \textbf{D=8} & \textbf{D=16} \\
\hline
Model Export    &       1.07     &      3.29     &      7.92     &      23.88    \\
Settings Generation     &       0.74     &      3.73     &      9.80     &      29.20    \\     
Circuit Compilation     &       0.31     &      2.31     &      7.12     &      24.20    \\      
Circuit Setup   &       155.04   &      156.06   &      174.11   &      176.53  \\
\hline
\end{tabular}
\caption{HAI Setup Time Breakdown (in seconds)}\label{tab:p1_st}
\end{table}

\begin{table}
\centering
\scriptsize
\setlength{\tabcolsep}{3pt}
\renewcommand{\arraystretch}{0.85}
\begin{tabular}{|l|c|c|c|c|}
\hline
\textbf{Setup} &  \multicolumn{4}{c|}{\textbf{Setup Time (s)}} \\
\cline{2-5}
\textbf{Phase} & \textbf{D=1} & \textbf{D=4} & \textbf{D=8} & \textbf{D=16} \\
\hline
Model Export    &       1.51     &      3.25     &      7.96     &      24.21    \\      
Settings Generation     &       0.74     &      3.53     &      9.55     &      28.99    \\     
Circuit Compilation     &       0.34     &      2.24     &      7.06     &      24.90    \\      
Circuit Setup   &       162.31   &      146.84   &      162.04   &      167.55 \\
\hline
\end{tabular}
\caption{ORNL-PS Setup Time Breakdown (in seconds)}\label{tab:ps_st}
\end{table}
The mean setup times for HAI and ORNL-PS dataset has been provided in Tables \ref{tab:p1_st} and \ref{tab:ps_st} respectively for PSF values of 8, 10 and 12. These tables report zkSTAR times for model export, settings generation, circuit compilation, and circuit setup on the HAI and ORNL-PS datasets. The results exhibit trends similar to Table~\ref{tab:tcws}, primarily reflecting circuit-size dependencies inherent to zk-SNARK setup. 


\begin{figure}[!htb]
    \centering
{\includegraphics[width=0.49\textwidth,keepaspectratio]{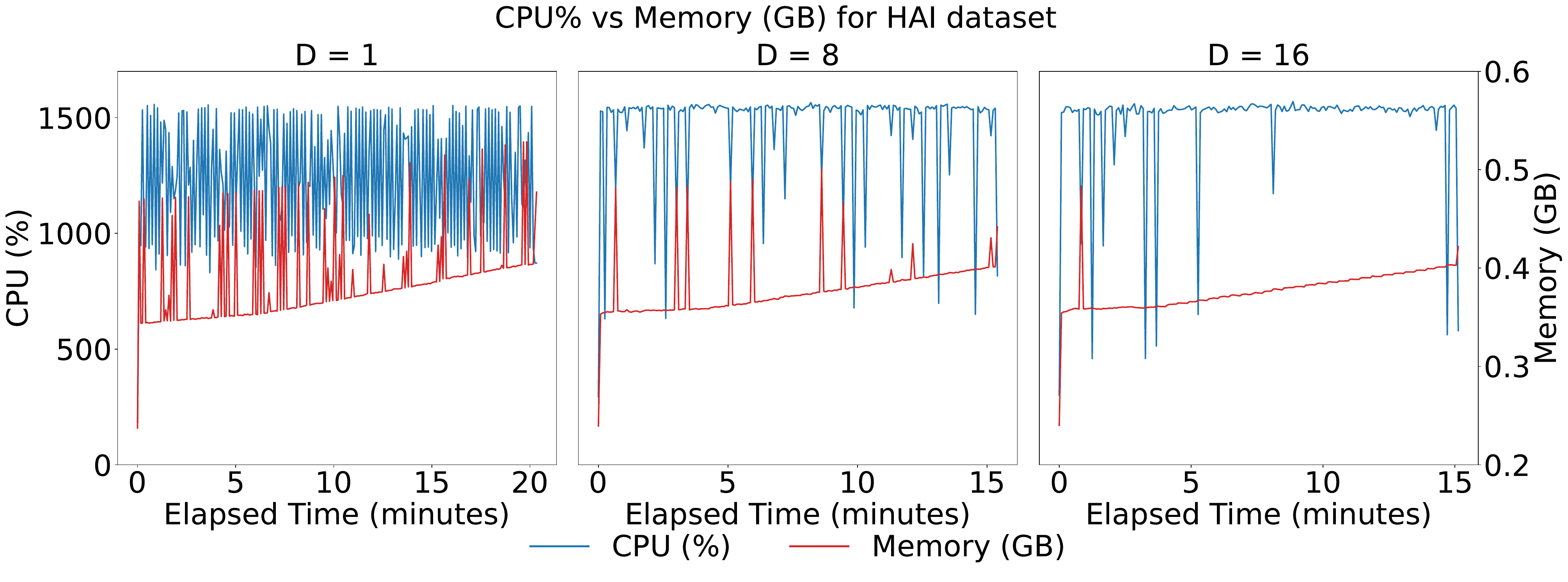}}\label{fig:wg_sysperf_hai}
  \caption{Witness Generation Performance with varying D for HAI dataset}\label{fig:p1_wgp}
\end{figure}
\begin{figure}[!htb]
    \centering
{\includegraphics[width=0.49\textwidth,keepaspectratio]{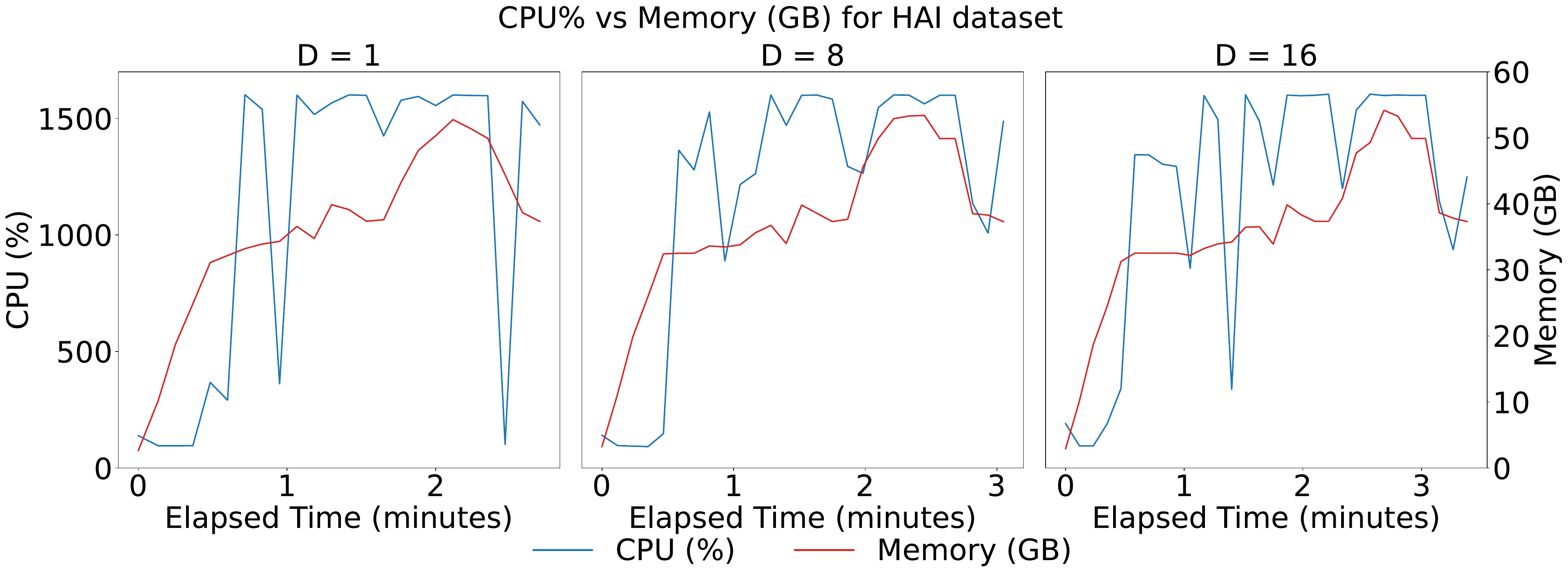}}\label{fig:pp_sysperf_hai}
  \caption{Prover Performance for Temporal Consistency  with varying D for HAI dataset}\label{fig:ps_wgp}
\end{figure}

\subsection{Witness and Proof Generation Performance}\label{witproof_gen}
Figures \ref{fig:p1_wgp} and \ref{fig:ps_wgp} present trends pertaining to the CPU utilization and memory consumption for only the HAI dataset. We provide the corresponding trends for ORNL-PS in Appendix \ref{subsec:ps_sysper_subsec} due to space constraints. The results show a steady and predictable scaling trend, with CPU utilization peaking around 1600\%, demonstrating effective parallelization, and memory usage reaching approximately 55~GB, well within available system capacity. From Figure \ref{fig:p1_wgp}, we can see greater volatility in CPU utilization and memory consumption for D value of 1 eventually stabilizing for D values of 8 and 16. The underlying reason for this trend is that lower values of D demand more frequent proving steps that results in higher volatility in CPU utilization and memory consumption. As D grows to 8 and 16, the proving time stabilizes on account of less frequent proving steps encountered during the simulation horizon. These results confirm that zkSTAR efficiently leverages multi-core resources with stability in memory consumption under increasing computational load.

\subsection{Verification Times}\label{subsec:verification_times}

\begin{table}[t]
\centering
\scriptsize
\setlength{\tabcolsep}{3pt}
\renewcommand{\arraystretch}{0.85}
\caption{Verification times for HAI and ORNL-PS datasets (PSF=10).}
\label{tab:verification_times}
\begin{tabular}{|c|c|c|c|c|c|}
\hline
\multicolumn{3}{|c|}{\textbf{HAI}} & \multicolumn{3}{c|}{\textbf{ORNL-PS}} \\ \hline
\textbf{D} & \textbf{Type} & \textbf{Time (s)} & \textbf{D} & \textbf{Type} & \textbf{Time (s)} \\ \hline
1  & prediction & 1.22 & 1  & prediction & 1.26 \\ 
8  & prediction & 1.34 & 8  & prediction & 1.40 \\ 
16 & prediction & 1.09 & 16 & prediction & 1.31 \\ 
16 & hypothesis  & 1.24 & 16 & hypothesis  & 1.46 \\ \hline
\end{tabular}
\end{table}
Table \ref{tab:verification_times} summarizes representative verification latencies for zkSTAR proofs under varying TC interval sizes of 1,8 and 16 for prediction and hypothesis test components.
Across all scenarios—including both normal and attack conditions—verification completed in approximately 1.0–1.5 s on commodity hardware, with all proofs passing successfully.
These results confirm that zkSTAR’s zero-knowledge verification incurs minimal overhead and is practical for near–real-time regulatory audits.

\end{document}